\documentclass[twoside, leqno, twocolumn]{article} 

\usepackage{pgfplots}
\usepackage{pgfplotstable}
\pgfplotsset{compat=1.8}
\usepackage{subcaption}
\usepackage{xspace} 
\usepackage{xcolor}
\usepackage{amsfonts}
\usepackage{algorithm}
\usepackage{algorithmicx}
\usepackage{float}
\usepackage{siunitx}
\usepackage{graphicx}
\usepackage{thmtools}
\usepackage{thm-restate}
\usepackage{blkarray}
\usepackage{balance}  
\usepackage{booktabs} 
\usepackage{multirow}

\usepackage{marginnote} 
\usepackage{url}
\usepackage{enumitem}
\graphicspath{{./fig/}}
\usepackage[noend]{algpseudocode}

\usepackage{comment}
\usepackage{color}
\usepackage{tikz}
\usepackage{todonotes} 
\usetikzlibrary{fadings}
\usepackage{minted}

\usepackage{dirtytalk}
\usepackage[flushleft]{threeparttable}
\usepackage{pgfplotstable}

\usepackage[T1]{fontenc}
\usepackage{siamproceedings}

\newcommand{\punt}[1]{}

\newtheorem{problem}{Problem}

\newtheorem{question}{Question}

\makeatletter
\def\@copyrightspace{\relax}
\makeatother

\newcommand{\defn}[1]       {{\textit{\textbf{\boldmath #1}}}}
\newcommand{\myparagraph}[1]{\vspace{.03in}\noindent {\textbf{\textit{#1}}}}

\newcommand{\polylog}{\mathrm{polylog}}

\newcommand{\prob}[1]{ \Pr \left [ #1 \right ]}

\renewcommand{\subparagraph}[1]{\smallskip
\noindent
\emph{#1 }}

\newcommand{\etal}{\text{et al}.\xspace}

\date{}

\newcommand{\namedcomment}[3]{{\sf \color{#2} #1: #3}}
\newcommand{\mab}[1]{\namedcomment{mab}{red}{#1}}
\newcommand{\mfc}[1]{\namedcomment{mfc}{purple}{#1}}
\newcommand{\david}[1]{\namedcomment{david}{red}{#1}}
\newcommand{\daniel}[1]{\namedcomment{daniel}{red}{#1}}
\newcommand{\ahmed}[1]{\namedcomment{ahmed}{blue}{#1}}
\newcommand{\victor}[1]{\namedcomment{victor}{green}{#1}}
\newcommand{\evan}[1]{\namedcomment{evan}{orange}{#1}}
\newcommand{\abi}[1]{\namedcomment{abi}{cyan}{#1}}

\renewcommand{\daniel}[1]{\todo[size=\tiny,color=cyan!40]{Dan: #1}}
\renewcommand{\mab}[1]{\todo[size=\tiny,color=green!40]{MAB: #1}}
\renewcommand{\mfc}[1]{\todo[size=\tiny,color=green!40]{MFC: #1}}
\renewcommand{\david}[1]{\todo[size=\tiny]{David: #1}}
\renewcommand{\ahmed}[1]{\todo[size=\tiny,color=yellow]{Ahmed: #1}}
\renewcommand{\victor}[1]{\todo[size=\tiny,color=yellow]{Victor: #1}}
\renewcommand{\evan}[1]{\todo[size=\tiny,color=red!40]{Evan: #1}}
\renewcommand{\abi}[1]{\todo[size=\tiny,color=red!40]{Abi: #1}}

\newcommand{\qnote}[1]{\todo[size=\tiny,color=red!40]{Quinten: #1}}

\newcommand{\fixme}[1]{\todo[size=\tiny]{#1}}
\newcommand{\inline}[1]{\todo[inline,color=yellow,size=\tiny]{#1}}

\renewcommand{\epsilon}{\varepsilon}

\renewcommand{\eqref}[1]          {Eq.~\ref{eq:#1}}

\definecolor{bg}{rgb}{0.95,0.95,0.95}

\usepackage[noEnd]{algpseudocodex} 
\usepackage{algorithm}

\algrenewcommand{\algorithmicindent}{1em}

\algblockdefx[PARFOR]{ParFor}{EndParFor}[1]{\textbf{for parallel }#1}{}

\DeclareCaptionFont{xbf}{\bfseries\boldmath}
\newcommand{\revision}[1]{#1}
\newcommand{\revise}[1]{#1}

\newcommand{\locked}{}

\newcommand{\graph}{\mathcal{G}}
\newcommand{\nodes}{\mathcal{V}}
\newcommand{\edges}{\mathcal{E}}
\newcommand{\tree}{\mathcal{T}}
\newcommand{\forest}{\mathcal{F}}
\newcommand{\component}{\mathcal{C}}
\newcommand{\nodesize}{V}
\newcommand{\edgesize}{E}
\newcommand{\graphstream}{S}

\newcommand{\edge}{e}

\newcommand{\nodesubset}{\mathcal{U}}

\newcommand{\lsketch}{s}

\newcommand{\myvec}{x}
\newcommand{\myothervec}{y}

\newcommand{\charvec}{f}

\newcommand{\hashfamily}{\mathcal{H}}

\newcommand{\losslessname}{lossless\xspace}
\newcommand{\Losslessname}{Lossless\xspace}
\newcommand{\sketchdcname}{sketch-based\xspace}
\newcommand{\Sketchdcname}{Sketch-Based\xspace}
\newcommand{\streamingname}{streaming connectivity\xspace}
\newcommand{\Streamingname}{Streaming Connectivity\xspace}

\newcommand{\cutset}{F}
\newcommand{\cutsetprob}{p}

\newcommand{\link}{\mathsf{Link}}
\newcommand{\cut}{\mathsf{Cut}}
\newcommand{\connected}{\mathsf{Connected}}
\newcommand{\update}{\mathsf{Update}}
\newcommand{\query}{\mathsf{Query}}
\newcommand{\pathquery}{\mathsf{PathQuery}}
\newcommand{\subtreequery}{\mathsf{SubtreeQuery}}

\newcommand{\queryperiod}{{\rho}}
\newcommand{\treeedge}{tree-edge\xspace}
\newcommand{\nontreeedge}{non-tree-edge\xspace}

\newcommand{\agmname}{AGM sketching\xspace}
\newcommand{\Agmname}{AGM Sketching\xspace}
\newcommand{\cameosketch}{\textsc{CameoSketch}\xspace}
\newcommand{\gibb}{Gibb's algorithm\xspace}
\newcommand{\Gibb}{Gibb's Algorithm\xspace}
\newcommand{\sysname}{CUPCaKE\xspace}
\newcommand{\sysacronym}{CUPCaKE\xspace}
\newcommand{\fullsysname}{\defn{\sysname} (\textbf{C}ompact \textbf{U}pdating \textbf{P}arallel \textbf{C}onnectivity \textbf{a}nd S\textbf{k}etching \textbf{E}ngine)\xspace}
\newcommand{\cdcname}{Concurrent-DC\xspace}
\newcommand{\cdcacronym}{CDC\xspace}
\newcommand{\dtree}{D-Tree\xspace}
\newcommand{\idtree}{ID-Tree\xspace}
\newcommand{\speculative}{speculative\xspace}
\newcommand{\speculatively}{speculatively\xspace}
\newcommand{\Speculative}{Speculative\xspace}

\newcommand{\geo}{\mathbb{G}}
\newcommand{\domain}{\mathcal{D}}

\begin{document}


\title{Fast and Compact \Sketchdcname Dynamic Connectivity}

\author{
Quinten De Man
\thanks{University of Maryland
(\email{deman@umd.edu}).}
\and Qamber Jafri
\thanks{Stony Brook University
(\email{qjafri28@gmail.com}).}
\and Daniel DeLayo
\thanks{Stony Brook University
(\email{ddelayo@cs.stonybrook.edu}).}
\and Evan T. West
\thanks{Stony Brook University
(\email{etwest@cs.stonybrook.edu}).}
\and David Tench
\thanks{Lawrence Berkeley National Lab
(\email{dtench@pm.me}).}
\and Michael A. Bender
\thanks{Stony Brook University
(\email{bender@cs.stonybrook.edu}).}
}

\date{}
\maketitle
\begin{abstract}

We study the dynamic connectivity problem for massive, dense graphs.
Our goal is to build a system for dense graphs that simultaneously answers connectivity queries quickly, maintains a fast update throughput, and a uses a small amount of memory. Existing systems at best achieve two of these three performance goals at once.

We present a parallel dynamic connectivity algorithm using graph sketching techniques that has space complexity $O(\nodesize \log^3 \nodesize)$ and query complexity $O(\log\nodesize/\log\log\nodesize)$.
Its updates are fast and parallel: in the worst case, it performs updates in $O(\log^2 \nodesize)$ depth and $O(\log^4 \nodesize)$ work. 
For updates which don't change the spanning forests maintained by our data structure, the update complexity is $O(\log \nodesize)$ depth and $O(\log^2 \nodesize)$ work.

We also present CUPCaKE (Compact Updating
Parallel Connectivity and Sketching Engine), a dynamic connectivity system based on our parallel algorithm.
It uses an order of magnitude less memory than the best lossless systems on dense graph inputs, answers queries with microsecond latency, and ingests millions of updates per second on dense graphs.

\end{abstract}

\sloppy
\clearpage

\section{Introduction}

The dynamic connected components problem is frequently encountered in a variety of application areas. The task is to compute the connected components of a graph subject to a sequence of edge insertions and deletions. This problem is motivated by applications such as tracking communities in social networks that change as users add or delete friends~\cite{dynamic_social,dynamic_web}, and identifying objects from a video feed rather than a static image~\cite{movingobject}. It is an important computational primitive for large-scale numerical analysis in computational geometry~\cite{doraiswamy2009efficient}, chemistry~\cite{eyal2005improved}, and biology~\cite{henzinger1999constructing}. It is frequently used as a subroutine in other graph algorithms (both static and dynamic), including dynamic minimum spanning forest and biconnectivity~\cite{conn_apps}. It is used as a heuristic in pathfinding algorithms (such as variants of Djikstra and A*) as well as for some clustering methods~\cite{zhang2020parallel,GeorganasEHGATBOY2018,nurk2017metaspades,azad2018hipmcl,van2000graph}.
 
While the details vary, all dynamic connectivity algorithms in the theory literature share two key features: they all construct a spanning forest of the graph to which new edges can be greedily added, and they all have a mechanism for replacing spanning forest edges that get deleted.
The task of replacing deleted edges is typically the most complicated and expensive part of these algorithms, and we categorize the known algorithms into two groups based on how they replace edges. \defn{\Losslessname algorithms}, starting with Henzinger and King~\cite{henzinger1995randomized}, store the current set of all edges in the graph losslessly, and must search this set to find edge replacements. On the other hand, \defn{\sketchdcname algorithms}, starting with Ahn \etal~\cite{Ahn2012}, compress the current set of all edges in the graph lossily into a data structure called a \defn{graph sketch}, which is asymptotically smaller than the input graph. The algorithm must perform a recovery procedure to sample replacement edges from the graph sketch.



\myparagraph{Dynamic connectivity challenge: What happens when the graphs are dense?}
In this paper we focus on solving dynamic connectivity specifically on \defn{dense} graphs. These graphs are challenging to process because their lossless representations are generally too large to fit in main memory. Before we explore solutions to this challenge, we first motivate the task of processing dense graphs.


\myparagraph{While most graphs are sparse, dense graphs do appear in important applications.}
The conventional wisdom is that massive graphs are always sparse, meaning that they have few edges per vertex.
Tench \etal~\cite{graphzeppelin, landscape} contend that instead large, dense graphs do not appear in academic publications due to a selection effect: since we lack the tools to process these graphs, they are not studied. Indeed, they point to reports of proprietary large, dense graphs used by companies like Facebook~\cite{facebookdense}, but these graphs are not openly available for public study. One such graph has 40 million nodes and 360 billion edges. In this paper we focus on methods for solving dynamic connectivity on dense graphs.

The practitioner who wants to solve dynamic connectivity on a dense graph faces a tradeoff---space versus query time. Lossless systems have massive space requirements on dense graphs but offer microsecond query latency. Graph sketching-based systems use small space on dense graphs but offer human-scale (single-digit second) query latency. 


\myparagraph{\Losslessname systems: huge space requirements but fast queries.} The best \losslessname algorithms~\cite{conn_apps,henzinger1997sampling,holm2001poly,wulff2013faster,thorup2000near,huang2023fully} have space complexity $O(\nodesize + \edgesize)$, worst-case query time complexity $O(\log \nodesize/\log\log\nodesize)$, and amortized update time complexity $o(\log^2 \nodesize)$ for graph $\graph = (\nodes, \edges), \nodesize = |\nodes|, \edgesize = |\edges|$.
Recent work has implemented modified versions of several \losslessname algorithms~\cite{fedorov2021scalable,chen2022dynamic, deman2025towards, xu2024constant, chen2025experimental}, and these implementations achieve high update throughput (hundreds of thousands to millions of updates per second) and low (microsecond) query latency. However, the linear dependency on $\edgesize$ in their space cost means that they cannot process large, dense graphs because the edge list is too large.

\myparagraph{Graph sketching systems: small space but slow queries.}
Connected components sketch algorithms~\cite{Ahn2012,graphzeppelin,landscape,bender2025case,gibb2015dynamic} have space complexity $O(\nodesize\log^3\nodesize)$ bits which is completely independent of $\edgesize$. As a result, sketching saves a lot of space when the input graph is dense, but does not save much space (and in fact may be larger than an adjacency list of $\graph$) when the graph is very sparse. Recent work has implemented connectivity sketching algorithms in the single-machine~\cite{tench2024graphzeppelin} and distributed~\cite{landscape} settings. These implementations are compact 
and have high update throughput 
on dense graphs. However, connected component queries are slow: the query complexity is $O(\nodesize\log^2\nodesize)$, and in practice this translates into human-scale (single-digit second) query latency.


Our goal is to have our cake and eat it too: that is, to build a system that uses small space on dense graphs, but also has fast queries.

Our starting point is the dynamic connectivity algorithm of Gibb \etal~\cite{gibb2015dynamic}, which has space complexity $O(\nodesize \log^3 \nodesize)$, update time complexity $O(\log^4\nodesize)$, and query complexity $O(\log\nodesize/\log\log\nodesize)$. Prior to this work, Gibb's algorithm had not been implemented to our knowledge. Implementing Gibb's algorithm alongside the sketching techniques of Tench \etal~\cite{graphzeppelin,landscape} immediately gives space almost as small as the best sketching systems, as well as microsecond-latency queries. Thus, for space and queries, an implementation of Gibb's is enough to declare victory. However, this implementation only achieved update throughput in the tens of thousands. There were two major obstacles we had to overcome to build a faster system based on Gibb's techniques: too many log factors in the update cost, and sequential bottlenecks.




For the first obstacle, Gibb's algorithm has update time $O(\log^4\nodesize)$, and the big $O$ notation hides a large constant. For comparison, state-of-the-art lossless dynamic connectivity systems are based on algorithms whose amortized update time is $O(\log^2\nodesize)$ or less. For a graph on a million nodes, this means that Gibb's could be roughly three orders of magnitude slower than the state-of-the-art lossless systems. 

For the second obstacle, graph sketching algorithms with much lower update costs than Gibb's algorithm require lots of parallelism to be fast in practice. Gibb's algorithm has lots of sequential data dependencies, so an implementation gains limited performance from parallelism. As a result, even though a system built on Gibb's algorithm has small space and fast queries, updates are three orders of magnitude slower than the fastest lossless systems.

Whereas before we described a two-way tradeoff between space and queries, in fact there is a three-way tradeoff between space, queries, and updates. Lossless systems have fast queries and updates but use lots of space. Sketching systems have fast updates and small space but slow queries. And an implementation of Gibb's algorithm has small space and fast queries but slow updates.

In this paper, we show how to get all three---small space, fast queries,  and fast updates---in one system.

\subsection{Results}

Our first result is a parallel dynamic connectivity algorithm using graph sketching which has space complexity $O(\nodesize \log^3 \nodesize)$ and query complexity $O(\log\nodesize/\log\log\nodesize)$.
Its updates are fast and parallel: in the worst case, it performs updates in $O(\log^2 \nodesize)$ depth and $O(\log^4 \nodesize)$ work. For a common class of updates called \defn{non-isolated}, which require minimal structural change in our data structure, the algorithm performs updates in only $O(\log \nodesize)$ depth and $O(\log^2 \nodesize)$ work. 

Our second result is \fullsysname\footnote{Our code can be found here: \url{https://anonymous.4open.science/r/DynamicQueriesSketch-9BC1/README.md}}, our dynamic connectivity system based on our parallel algorithm. \sysname takes in a stream of edge insertions and deletions that define a graph and supports connectivity queries at any point during the stream. It is compact: it uses an order of magnitude less memory than the best lossless systems on dense graph inputs. It answers queries quickly: on a variety of datasets, queries have microsecond latency. It has reasonably fast updates: it ingests updates up to two orders of magnitude faster than Gibb's algorithm. On dense graphs, \sysname is 1.5-8 times faster than nearly all state-of-the-art lossless dynamic connectivity systems~\cite{fedorov2021scalable, chen2022dynamic}, with the exception of the current fastest system~\cite{xu2024constant}, which is about ten times faster than it.

\myparagraph{Our Algorithm Has Low Costs in the Worst-Case and in the Common-Case.}
Our parallel algorithm achieves $O(\log^2 \nodesize)$ worst-case depth in the binary fork-join model~\cite{binaryforking}.
This is the lowest worst-case update depth of any known dynamic connectivity algorithm, improving on the previous best known result of $O(\log^3 \nodesize)$ worst-case depth~\cite{acar2019parallel,deman2025towards}.

It also performs asymptotically better on a common class of updates that we call \defn{non-isolated} (defined in Section~\ref{sec:algorithm}). Roughly, non-isolated updates are ones that do not change the spanning forests maintained by our data structure. Experimentally, we find that nearly all updates ($93\% - 99.994\%$, depending on the dataset) are non-isolated.

The improved parallelism (lower depth) and low asymptotic work for updates allows our algorithm to overcome the two obstacles that prevent an implementation of \gibb from having fast updates: sequential data dependencies and high average work per update.

\myparagraph{\sysname: High-Performance Dynamic Connectivity Implementation}
\sysname is based on our parallel dynamic connectivity algorithm. We now summarize three design choices and heuristics that allow \sysname to achieve small space, fast queries, and fast updates in experiments.
First, \sysname adopts a \defn{partitioned memory design} in order to minimize communication in the easy case (non-isolated updates) and decreases cache contention in the hard case (isolated updates).
Second, \sysname uses a heuristic we call \defn{\speculative update buffering} that takes advantage of non-isolated update's greatly reduced data-dependency to process many of them in parallel. 
%
Third, \sysname's dynamic tree data structures are built on what we call a \defn{reduced height skip list} which drastically decreases the work of non-isolated updates and significantly increases query speed.

\section{Preliminaries}
\subsection{Background \locked}
\label{subsec:prelimsummary}
This paper assumes the reader has some familiarity with dynamic connectivity algorithms and graph sketching algorithms. We provide a brief summary of several important concepts here; for a more detailed treatment on each of the topics summarized here, we refer the curious reader to Appendix~\ref{app:prelims}.

We consider the dynamic connectivity problem, where the task is to answer connectivity queries on a graph graph $\graph = (\nodes, \edges)$ defined by a series of edge insertions and deletions. Queries have form $\connected(u,v)$ and should return YES if there is a path between vertices $u,v \in \nodes$ and NO otherwise. Queries arrive interleaved with updates and must be answered before processing more updates. The goal is to minimize update time, query time, and space cost.

Variants of this problem have been studied in different algorithmic settings. \Losslessname graph algorithms~\cite{conn_apps,henzinger1997sampling,holm2001poly,wulff2013faster,thorup2000near,huang2023fully} and systems~\cite{fedorov2021scalable,chen2022dynamic, deman2025towards, xu2024constant, chen2025experimental} focus on maximizing update throughput and minimizing query latency. Graph sketching algorithms~\cite{Ahn2012,bender2025case} and systems~\cite{graphzeppelin,landscape} focus on minimizing space cost as a first-order concern. 

Graph sketching algorithms minimize space by representing parts of the input graph $\graph$ in a lossily compressed form. In this work we make use of an algorithmic primitive used in most graph sketching algorithms called an \defn{$\ell_0$ sketch}~\cite{cormode2014unifying,Ahn2012}. $\ell_0(u)$, the $\ell_0$ sketch of $u \in \nodes$, is a data structure of size $O(\log^2 \nodesize)$ bits which lossily represents the neighborhood of $u$. An insertion or deletion of edge $(u,v)$ can be applied to $\ell_0(u)$ in constant time~\cite{landscape}. At any time, querying $\ell_0(u)$ gives an edge incident to $u$ with constant probability, and otherwise fails to return an edge. Addition is defined over $\ell_0$ sketches; for $S \subset \nodes$, we denote $\sum_{u\in S} \ell_0(u)$ as $\ell_0(S)$. Querying $\ell_0(S)$ gives an edge $(w,y)$ s.t. $w \in S$ and $y \in \nodes \setminus S$ with constant probability, and otherwise fails to return an edge. 

\subsection{\Sketchdcname Dynamic Connectivity \locked}\label{sec:hybrid_solutions}
More recently, a new class of dynamic connectivity algorithms has emerged which combines both the polylogarithmic update and query times of \losslessname algorithms and the small space usage of \streamingname algorithms~\cite{kapron2013dynamic, gibb2015dynamic}. We call these \defn{\sketchdcname dynamic connectivity algorithms}.
The new algorithms in this paper share several characteristics of the \sketchdcname dynamic connectivity algorithm of Gibb~\etal~\cite{gibb2015dynamic}, so we summarize it here.
The algorithm uses $O(\nodesize \log^3 \nodesize)$ bits of space, processes updates in $O(\log^4 \nodesize)$ worst-case time, and answers queries in $O(\log \nodesize / \log \log \nodesize)$ time. Each query is correct w.h.p.
We refer to their algorithm as \gibb.

%

\subsubsection{Cutset Data Structure \locked}
The main components in \gibb are \defn{cutset} data structures~\cite{kapron2013dynamic}. A cutset for a graph $\graph = (\nodes, \edges)$ and a forest of the graph $\forest = (\nodes, \edges_{\forest} \subseteq \edges)$ is able to find an edge crossing the cut of components in the graph quickly, even as edges are inserted to or deleted from $\graph$ and components in $\forest$ are linked or cut.
Let $\component(w)$ denote the component of $\forest$ containing a vertex $w$. Formally, a cutset supports the following operations:

\begin{itemize}[topsep=0pt,itemsep=0pt,parsep=0pt,leftmargin=15pt]
    \item $\link(u,v)$: Given vertices $u$ and $v$ that are not connected in $\forest$, insert the edge $(u,v)$ into $\forest$, linking $\component(u)$ and $\component(v)$.
    \item $\cut(u,v)$: Given vertices $u$ and $v$ where $(u,v)$ is an edge in $\edges_{\forest}$, delete the edge, cutting $\component(u,v)$ into $\component(u)$ and $\component(v)$.
    \item $\update(e)$: If edge $e \notin \edges$, insert $e$ into $\edges$, else delete $e$ from $\edges$.
    \item $\query(v)$: Return an edge crossing the cut of $\component(v)$. This is successful with at least some constant probability $0 < \cutsetprob < 1$.
\end{itemize}
The cutset data structure can be implemented as a \defn{dynamic tree} data structure, which maintains a forest that represents the components of the graph subject to edge insertions and deletions ($\link$ and $\cut$ operations) on a fixed vertex set.
The dynamic tree data structure maintains a value associated with each vertex and can return the value of an aggregate function $g$ applied over the values for all the vertices in a target component (see Section~\ref{app:dynamic_trees} for details).
To implement the cutset, the dynamic tree maintains $\ell_0$-sketch $\ell_0(v)$ with success probability $\delta = 1 - \cutsetprob$ as the value for each vertex $v \in \nodes$ (see Section~\ref{sec:streaming_solutions} for more details).
Specifically, the cutset algorithm uses the $\ell_0$-sketch algorithm of Cormode~\etal~\cite{cormode2014unifying}.
The aggregate function $g$ is simply the addition function for two $\ell_0$-sketches. 

Performing a $\link(u,v)$ or $\cut(u,v)$ operation in the cutset corresponds exactly to the same operation in the dynamic tree data structure.
An $\update(e = (u,v))$ operation first performs sketch updates to $\ell_0(u)$ and $\ell_0(v)$. The sketch update must also propagate to $O(\log \nodesize)$ internal nodes in the dynamic tree (see Appendix~\ref{app:dynamic_trees}).
A $\query(v)$ operation entails performing a $\subtreequery(v)$ operation in the dynamic tree data structure to get the aggregate value for the component containing $v$, which we call $\component(v)$. This aggregate value is the sketch $\sum_{v \in \component} \ell_0(v)$.
Thus, the aggregate sketch supplied by the tree returns an edge crossing the cut $(\component(v), \nodes \setminus \component(v))$ with probability at least $p$.
%
Lemma~\ref{lem:cutset} analyzes the cutset data structure:
\begin{lemma}[Adapted from~\cite{gibb2015dynamic}, Lemma 2.1] \label{lem:cutset}
    For a graph $\graph = (\nodes, \edges)$ and a forest $\forest = (\nodes, \edges_{\forest} \subseteq \edges)$ there exists a cutset data structure with constant query success probability $p$ using $O(\nodesize \log^2\nodesize)$ space, performing $\link$ and $\cut$ operations in $O(\log^2 \nodesize)$ time, $\update$ operations in $O(\log^2\nodesize)$ time, and $\query$ operations in $O(\log \nodesize)$ time.
\end{lemma}

\subsubsection{\Gibb \locked}
Gibb's algorithm~\cite{gibb2015dynamic} maintains cutset data structures $\cutset_0, \hdots, \cutset_{top}$ on the graph $\graph$ (where $top = O(\log \nodesize)$) for $O(\log \nodesize)$ levels of forests $\forest_0, \hdots, \forest_{top}$ where the forests at higher levels represent successively larger partial connected components, and eventually represent the true components. 
Additionally, a separate dynamic tree $\tree$ is kept which supports path queries for the maximum weight edge (see Appendix~\ref{app:dynamic_trees}). The edges in this tree are identical to $\edges_{\forest_{top}}$ and the weight of each edge $\edge$ is defined to be the lowest level $i$ such that $\edge \in \edges_{\forest_{i}}$.
Formally, the algorithm maintains the following invariants, giving Lemma~\ref{lem:gibb_invs}:
\begin{enumerate}[topsep=0pt,itemsep=0pt,parsep=0pt,leftmargin=25pt]
    \item \label{inv:gibb1} $\forest_0 = (\nodes, \emptyset)$ is just the vertices of $\graph$.
    \item \label{inv:gibb2} For each level $0 \leq i < top, \forest_i \subseteq \forest_{i+1}$.
    \item \label{inv:gibb3} For each level $0 \leq i < top$, for each component $\component \in \forest_i$, if a query on $\cutset_i$ for component $\component$ would be successful, then $\component$ is a proper subset of some component $\component' \in \forest_{i+1}$.
\end{enumerate}

\begin{lemma}[Adapted from~\cite{gibb2015dynamic}, Lemma 3.3] \label{lem:gibb_invs}
    For any graph $\graph$ and any constant $c$, let $p$ be the probability of success for the cutset data structure, let $\alpha = \lceil \log_{4/(4-p)} \nodesize \rceil$, and let $top = \max\{2\alpha/\beta,8c\ln \nodesize/\beta\}$ where $\beta = (1-p)/(1-p/2)$.
    If Invariants~\ref{inv:gibb1}--\ref{inv:gibb3} hold for the data structure, then $\forest_{top}$ is a spanning forest of $\graph$ with probability at least $1-1/\nodesize^c$.
\end{lemma}

Appendix~\ref{app:gibb} contains a detailed description and illustrative example of the update and query procedures for \gibb.
Gibb~\etal prove that for a single update, the invariants are maintained w.h.p. and they bound the total update time by $O(\log^4 \nodesize)$. Thus the answers to connectivity queries are correct w.h.p. across a polynomial number of updates. They also show that connectivity queries can be answered in $O(\log \nodesize / \log \log \nodesize)$ time.

\section{Parallel \Sketchdcname Dynamic Connectivity \locked}\label{sec:algorithm}
Towards the goal of a dynamic connectivity system with small space, fast queries, and reasonable update throughput, we present a parallel \sketchdcname dynamic connectivity algorithm which has low span per update. This algorithm forms the basis for our practical system \sysname (see Section~\ref{sec:perf_eng}).

In the interest of designing a performant algorithm, we categorize dynamic connectivity updates in \gibb into two cases based on their impact on the cutset data structure.
We define an \defn{isolated update} as an update during which any $\link$ or $\cut$ operation is induced in any cutset.
We define a \defn{normal update} as one that induces no $\link$ or $\cut$ operations.
These definitions naturally extend to the algorithm we propose later in this section, which also make use of cutset data structures.
Since normal updates don't induce structural changes to the spanning forest at any level, they are considered the ``easy'' case.

Experimentally, we found that isolated updates are rare in our data (see Appendix~\ref{app:isolated}), indicating that it would be promising to design an algorithm that performs exceptionally well on normal updates.
Thus, our new algorithm is designed to have a low asymptotic cost for normal updates.
In Appendix~\ref{app:gibb} we prove that \gibb performs normal updates in $O(\log^3 \nodesize)$ time.
In contrast, our new parallel algorithm in this section performs normal updates in only $O(\log^2 \nodesize)$ work and $O(\log \nodesize)$ depth.
In Section~\ref{sec:gibb_results} we experimentally compare our algorithm against \gibb, and find that our algorithm is significantly faster.

\subsection{Parallel Cutset \locked} \label{sec:par_cutset}

We present a \defn{parallel cutset data structure} and analyze it in the binary fork-join model~\cite{binaryforking}.
Our parallel cutset achieves parallelism in $\link$ and $\cut$ operations and achieves better (sequential) time complexity for $\update$ operations compared to the cutset described in Section~\ref{sec:hybrid_solutions}.
The first straightforward improvement over the sequential cutset data structure is to use \cameosketch~\cite{landscape} rather than the sketch of Cormode~\etal~\cite{cormode2014unifying}, reducing the time for sketch updates from $O(\log \nodesize)$ to $O(1)$, and reducing the time of the cutset $\update$ from $O(\log^2\nodesize)$ to $O(\log \nodesize)$.

Now we describe our parallel algorithms for $\link$ and $\cut$ operations.
Recall that the cutset data structure is implemented as a sketch-augmented dynamic tree (see Appendix~\ref{app:dynamic_trees} for details about dynamic trees).
Our algorithm is independent of the specific dynamic tree data structure used, and modifies the sequential $\link$ or $\cut$ algorithm as follows.
The operation proceeds as normal, except that whenever it would perform a sketch addition, instead of executing the addition it logs the pair of sketches to add and where to write the result. These addition tasks are stored in an array in the order the algorithm encounters them. 
This part takes $O(\log\nodesize)$ depth and work.

The next phase of the algorithm will employ parallelism to execute all the sketch additions.
A sketch is essentially an array of $O(\log \nodesize)$ machine words, and a sketch addition simply adds each array entry to the corresponding entry in the other sketch (see Appendix~\ref{app:prelims} for more detail).
Our algorithm forks into $O(\log\nodesize)$ threads (in $O(\log\log \nodesize)$ depth in the binary fork-join model), one for each index in the array of machine words per sketch. Each thread then asynchronously follows the sequence of stored sketch additions and executes the addition only for the index in the sketch array which that thread is responsible for.
Since there are only $O(\log V)$ sketch additions in the sequence, the depth of this part is $O(\log V)$, and the work is $O(\log^2\nodesize)$.
Finally, all of the threads are joined in $O(\log\log\nodesize)$ depth, giving us the following lemma:
%
\begin{lemma}\label{lem:par_cutset}
    For a graph $\graph = (\nodes, \edges)$ and a forest $\forest = (\nodes, \edges_{\forest} \subseteq \edges)$ there exists a cutset data structure that performs $\link$ and $\cut$ operations in $O(\log \nodesize)$ depth and $O(\log^2 \nodesize)$ work in the binary fork-join model, performs $\update$ operations in $O(\log \nodesize)$ time, and performs $\query$ operations in $O(\log \nodesize)$ time.
\end{lemma}

\subsection{Parallel Connectivity Algorithm \locked}\label{sec:par_conn}
Here we describe our parallel algorithm for sketch-based dynamic connectivity in the binary fork-join model~\cite{binaryforking}.
Our starting point is similar to that of \gibb: we maintain $O(\log \nodesize)$ levels of cutset data structures, this time parallel cutsets. Our algorithm also maintains Invariants~\ref{inv:gibb1}--\ref{inv:gibb3}.
Algorithm~\ref{alg:par_update} shows the pseudo-code for our parallel algorithm for edge insertions and deletions.

\begin{algorithm}[ht]
\caption{$\mathsf{ParallelUpdate}(e=(u,v), \text{bool } del)$}
\label{alg:par_update}
\begin{algorithmic}[1]
    \For {\textbf{parallel} $i \in [0,top]$} \label{line:par_update_start}
        \State $\cutset_i.\update(e)$
        \If {$del \land e \in \edges_{\forest_i}$} $\cutset_i.\cut(u,v)$ \EndIf
    \EndFor \label{line:par_update_end}

    \Statex
    \For {\textbf{parallel} $i \in [0,top-1]$} \label{line:par_iso_start}
        \State $violation_i \gets \infty$
        \For {$w \in \{u,v\}$}
            \State $\component_i \gets$ level $i$ component containing $w$
            \State $\component_{i+1} \gets $ level $i+1$ component containing $w$
            \State $\edge_{link} \gets \cutset_i.\query(w)$
            \If{$(\edge_{link} \neq \emptyset) \land (\component_i  = \component_{i+1})$}
            \State $violation_i \gets i$ \EndIf
        \EndFor
    \EndFor \label{line:par_iso_end}

    \Statex
    \State $min_\ell \gets \mathsf{ParallelMin}_{i=0}^{top-1} (violation_i)$ \label{line:par_minimum}
    \If{$min_\ell < \infty$} \Return \EndIf
    
    \Statex
    \For {$i \in [min_\ell,top-1]$} \label{line:par_main_loop_start}
    \For {$w \in \{u,v\}$}
        \State $\component_i \gets$ level $i$ component containing $w$
        \State $\component_{i+1} \gets $ level $i+1$ component containing $w$
        \State $(u_{link}, v_{link}) \gets \cutset_i.\query(w)$
        \If{$((u_{link}, v_{link}) \neq \emptyset) \land (\component_i = \component_{i+1})$}
            \If{$\tree.\connected(u_{link}, v_{link})$}
                \State ($u_{cut},v_{cut},\ell) \gets \tree.\pathquery(u_{link},v_{link})$
                \For {\textbf{parallel} $j \in [\ell,top]$} \label{line:par_cut_start}
                    \State $\cutset_j.\cut(u_{cut},v_{cut})$ \label{line:par_cut_end}
                \EndFor
            \EndIf
            \For {\textbf{parallel} $j \in [i+1,top]$} \label{line:par_link_start}
                \State $\cutset_j.\link(u_{link},v_{link})$ \label{line:par_link_end}
            \EndFor
        \EndIf
    \EndFor
    \EndFor \label{line:par_main_loop_end}
\end{algorithmic}
\end{algorithm}

First, each cutset calls $\update$ on the updated edge (parallel for loop on lines~\ref{line:par_update_start}--\ref{line:par_update_end}). This can be done in parallel for each level since the updates are independent of each other.
Additionally, if the update is an edge deletion and $e \in \edges_{\forest_i}$ then $\cut(u,v)$ is called on $\cutset_i$ (line~\ref{line:gibb_update_cut}).
In the binary fork-join model, the threads for all iteration in the parallel for loop can be spawned (and later joined) in $O(\log\log\nodesize)$ depth. Each iteration of the loop is a single $\update$ operation in $O(\log\nodesize)$ work.
Each $\cut$ has $O(\log \nodesize)$ depth and $O(\log^2 \nodesize)$ work, so overall this part takes $O(\log\nodesize)$ depth and $O(\log^3\nodesize)$ work.

Next, the algorithm checks for components that violate Invariant~\ref{inv:gibb3} at each level.
Importantly, as in \gibb~\cite{gibb2015dynamic}, during an update a component can only violate Invariant~\ref{inv:gibb3} if it contains $u$ or $v$.
A parallel for loop is used with one iteration for each level (lines~\ref{line:par_iso_start}--\ref{line:par_iso_end}).
The thread for each level locally computes whether the components containing $u$ and $v$ violate Invariant~\ref{inv:gibb3}.
Checking if a component violates Invariant~\ref{inv:gibb3} can be done by querying $\cutset_i$ with vertices $u$ and $v$. If the query is successful (returning an edge $e_{link}=(u_{link},v_{link})$) and $\component_{i+1}$ is equal to $\component_i$ (equality can be checked by comparing the size of the components which can be obtained from the dynamic tree data structure), then $\component_i$ violates Invariant~\ref{inv:gibb3}.
This part of the algorithm takes $O(\log\nodesize)$ time per loop iteration for a total of $O(\log\nodesize)$ depth and $O(\log^2\nodesize)$ work.
Once the threads are all joined, a parallel minimum operation is used to determine whether any level had a component violating Invariant~\ref{inv:gibb3} (line~\ref{line:par_minimum}).
If no violations of the invariants were found, the update is complete and the algorithm can terminate early.
Normal updates would be complete at this point, and the work and depth are $O(\log^2 \nodesize)$ and $O(\log \nodesize)$ respectively.

Else the algorithm must proceed to restore the invariants.
Starting from the lowest level with a violating component, level $min_\ell$, to the $top-1$ level, we use a parallel version of the process from \gibb for restoring the invariants (lines~\ref{line:par_main_loop_start}--\ref{line:par_main_loop_end}).
At each level we check whether the components containing $u$ and $v$ violate Invariant~\ref{inv:gibb3}. If it does, it returns an edge $e_{link}=(u_{link},v_{link})$. The algorithm should call $\link(u_{link},v_{link})$ at all levels above $min_\ell$ to restore Invariant~\ref{inv:gibb3} while maintaining Invariant~\ref{inv:gibb2} (lines~\ref{line:par_link_start}--\ref{line:par_link_end}).

However, it may not be possible to immediately link the edge $\edge_{link}$ in $\forest_j$ if $u_{link}$ and $v_{link}$ are already connected in $\forest_j$, because then this would form a cycle. So it is necessary to first find the minimum level $\ell > i$ where $u_{link}$ is connected to $v_{link}$, and disconnect them in all levels $\geq \ell$ (lines~\ref{line:par_cut_start}--\ref{line:par_cut_end}).
The lowest level where two vertices are connected corresponds to the maximum weight edge on their path in $\tree$.
To find this level (and a corresponding edge to cut), the algorithm calls $\tree.\pathquery(u_{link},v_{link})$.

The set of $\link$ operations and set of $\cut$ operations are performed in parallel for each level, and use the parallel cutset $\link$ and $\cut$ algorithms described earlier in this section.
Each $\link$ or $\cut$ takes $O(\log\nodesize)$ depth and $O(\log^2\nodesize)$ work. So each time these parallel for loops are executed it adds $O(\log\nodesize)$ depth and $O(\log^3\nodesize)$ work.
These loops happen at most once per iteration in the main for loop, so in the worst case an update takes $O(\log^2\nodesize)$ depth and $O(\log^4\nodesize)$ work.
Our parallel algorithm yields the following theorem:
\begin{theorem}\label{thm:par_alg}
    For a graph $\graph=(\nodes,\edges)$, there exists a dynamic connectivity algorithm that performs isolated updates in $O(\log^2 \nodesize)$ depth and $O(\log^4 \nodesize)$ work and performs normal updates in $O(\log \nodesize)$ depth and $O(\log^2 \nodesize)$ work in the binary fork-join model.
\end{theorem}

This algorithm is interesting because it has the lowest update depth of any algorithm in the dynamic connectivity literature (improving on the best known result of $O(\log^3\nodesize)$ depth~\cite{acar2019parallel, deman2025towards} in the strictly stronger PRAM model).
It also forms the basis for our practical system \sysname, described in Section~\ref{sec:perf_eng}.

\section{System Design} \label{sec:perf_eng}

In this section we describe 
\fullsysname, our system based on the parallel \sketchdcname dynamic connectivity algorithm from Section~\ref{sec:algorithm}. Our algorithm is implemented for a single multi-core shared memory machine.
We describe our system and several heuristics and design choices that are crucial for achieving good performance.

\subsection{Implementation Details \locked}\label{sec:implementation_details}

We developed our own implementation of Euler tour trees which uses a skip list as the underlying data structure, maintains an $\ell_0$-sketch augmented value for each vertex, and supports querying for the sum of all sketches in a target component.
%
For $\ell_0$-sketching use the \cameosketch code from Landscape~\cite{landscape}.
To optimize our memory usage,
we only maintain a sketch for a single instance of each vertex in the Euler tour.
This allows us to avoid storing approximately half of the sketches in the bottom level of each skip list. This is significant because the majority of nodes are in the bottom level of a skip list and the sketches asymptotically dominate the memory usage.

In our skip list implementation we randomly select the height of each element from a geometric distribution $\geo(1 - 1 / \log \nodesize)$ (see Appendix~\ref{app:skiplist} for details about skip lists). We refer to this as a \defn{reduced height skip list}. Intuitively, this heuristic leads to skip lists with lower height, meaning an $\update$ operation affects less sketches.
Additionally, we maintain ``parent pointers'' for each skip list node. We define the \defn{parent} of a level $i$ skip list node $w$ as the rightmost level $i+1$ node whose corresponding level $i$ node is to the left of $w$.
In Appendix~\ref{app:skiplist} we prove that a reduced height skip list with parent pointers has height $O(\log \nodesize / \log\log \nodesize)$ w.h.p., and any search path has length $O(\log^2 \nodesize)$ w.h.p. When applied to the cutset data structure, this allows for $\update$ in $O(\log \nodesize / \log \log \nodesize)$ time, while the time for $\link$ and $\cut$ operations increases by at most an $O(\log \nodesize)$ factor.
We experimentally evaluate the impact of this optimization in Section~\ref{sec:gibb_results}, and found it to highly beneficial for all of our datasets.

We also developed our own implementation of link-cut trees supporting path maximum queries using splay trees as the underlying data structure.
To answer connectivity queries we use an Euler tour tree built on a reduced height skip list that is not augmented with sketches. We found this to perform much better than using the link-cut tree for queries. Since there is no sketch augmentation, the links and cuts to this query Euler tour tree are fast.

\subsection{Partitioned Memory Design \locked}
We partition the memory used for \sysname into $top = O(\log\nodesize)$ parts and assign a process to each of these parts. Process $i$ is responsible for maintaining the $i$-th level of the dynamic connectivity data structure. That is, it maintains the parallel cutset data structure for level $i$ which is made up of Euler tour trees. The $top$ level is the only exception, because it instead contains the sketchless Euler tour tree used for connectivity queries and the link-cut tree.
Processes cannot access the memory assigned to other processes directly, and communicate with each other only using MPI. The $top$ level receives the input stream and connectivity queries.
Partitioning the levels of the data structure in this way minimizes communication in the easy case (normal updates) and decreases cache contention in the hard case (isolated updates).
For queries, as long as all pending updates have been processed, no communication needs to occur between processes because the query can be answered by the Euler tour tree at level $top$.

For convenience of presentation, we use a slightly different definition of an isolated update in this Section and Section~\ref{sec:update_buffer}. We call a component that violates Invariant~\ref{inv:gibb3} an \defn{isolated component}. We define an \defn{isolated update} to be an update for which there is any isolated component. This only differs from the previous definition by excluding updates that cut a current spanning forest edge, but have no isolated components.

\begin{figure*}[ht]
    \centering
\includegraphics[width=.8\linewidth]{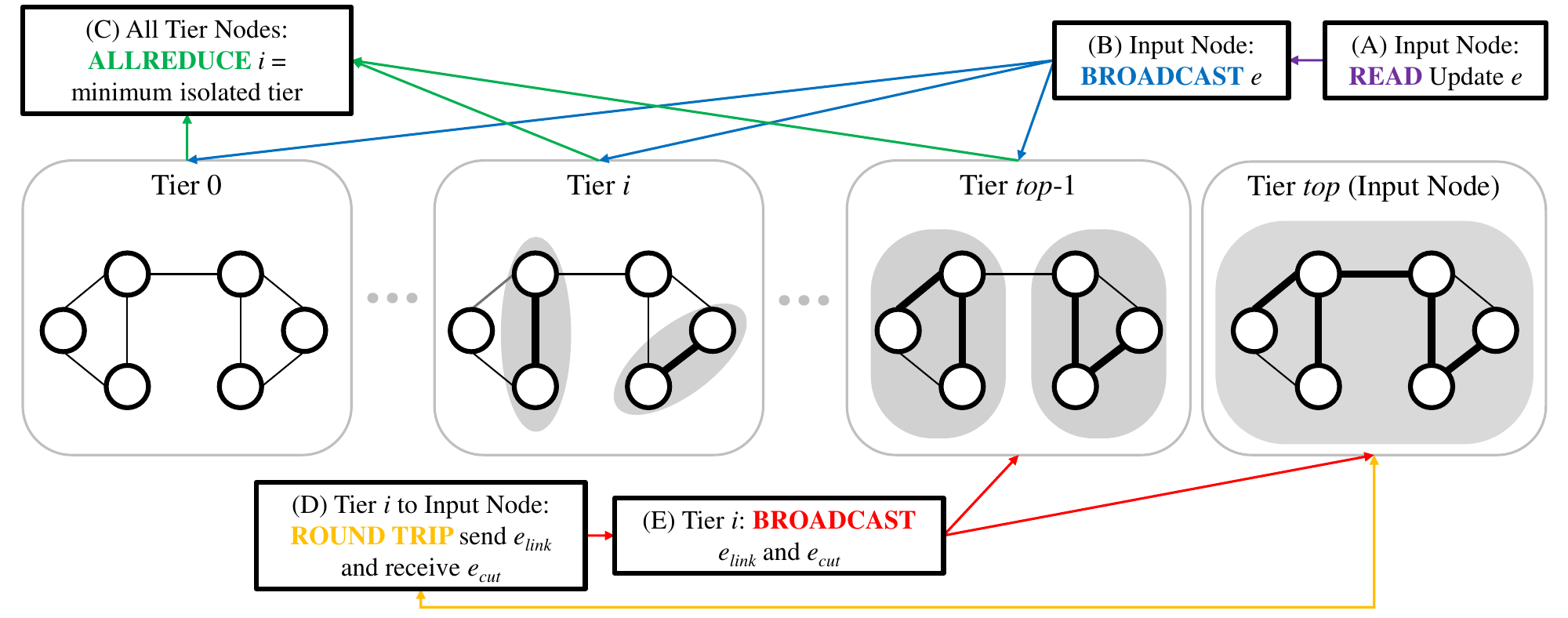}
    \caption{\small
    We illustrate our partitioned memory design using the example from Figure~\ref{fig:gibb_update}. Inter-tier communication is used to inform tiers of a new edge, identify isolated components, and insert/remove tree edges.
    First, (A) an update is read at the top tier and (B) the update is broadcasted to all other tiers. All tiers then locally update their cutset. Next, (C) an allreduce communication is used to determine the minimum isolated tier $i$. Then, (D) tier $i$ communicates with the $top$ tier to determine if there is a cycle edge to cut. Finally, (E) the new edge to link and the cycle edge to cut are broadcasted to all tiers $> i$. Steps C, D, and E are repeated until there is no isolated tier.
    }
    \label{fig:alg}
\end{figure*}

Figure~\ref{fig:alg} illustrates the inter-level data flow for normal and isolated updates.
In an update, (1) the edge $e$ is sent from the input stream to the $top$ process.
Then, (2) the $top$ process broadcasts $e$ to all processes, which perform process-local work like performing $\update(e)$ on their cutset.
Next, (3) each process determines whether its cutset has any isolated components, and an AllReduce operation determines the minimum process $i$ that is isolated. If there was no isolated component, the update was normal, and the update is complete.
For an isolated update, (4) process $i$ samples an edge (which we call $e_{link}$) from its isolated component and sends $e_{link}$ to the $top$ process, which responds with the lowest-weight edge in the resulting cycle $e_{cut}$ (if such an edge exists).
Finally, (5) process $i$ broadcasts $e_{link}$ and $e_{cut}$ to all processes. Now process $i$ is not isolated but higher processes may still be, so return to step 3 and continue until no processes are isolated.
%
%
Overall there are at most $O(\log \nodesize)$ rounds of communication, each with a constant amount of data transferred.

\subsection{\Speculative Graph Update Buffering \locked}\label{sec:update_buffer}
Here we describe a technique that takes advantage of potential parallelism between groups of consecutive normal updates, and reduces the frequency of communication rounds in our partitioned memory design.
In our current design, each update begins by broadcasting the update to all of the tiers. If the update was normal then each tier can handle the update independently and no further communication is needed.
Similarly, if we know ahead of time that there will be a group of several consecutive normal updates, they may all be buffered and broadcast in a single round of communication.

Our technique is to \defn{\speculatively buffer} a small number of updates hoping that they will all be normal updates. Then each of these updates will be partially processed.
Specifically, the group of updates is first broadcast to each tier, and the initial $\update$ and $\cut$ operations are applied for each update in the group.
Next, each tier finds the earliest update for which a component containing one of its endpoints is isolated.
Then, a second round of communication occurs to find the earliest update whose endpoints are contained in a violating component on any tier. We call this the \defn{earliest isolated update}. 
If there were no isolated updates, then the updates are done.

If there was an isolated update, then some of the \speculative work done from later updates in the group must be reverted. Specifically, the operations done for all of the updates after the earliest isolated update are undone.
For $\update$ operations, executing the same operation again cancels out the previous execution.
For $\cut$ operations, we must revert them with an $\link$ operation.
This returns the state of the data structure to what it would be if only the updates before the earliest isolated update were performed.
Then the regular algorithm algorithm is used to sequentially process the earliest isolated update and all following updates.
%
%
%
For queries, all of the updates that we have \speculatively buffered must be immediately processed before answering the query.
We experimentally evaluate the impact of buffer size on update throughput and query latency in Appendix~\ref{app:update_buffer_experiment}.

\section{Experimental Setup \locked}\label{sec:experimental_setup}

We implemented \sysname as a C++17 executable compiled with Open MPI version 5.0.2 and GCC version 13.2.0. All experiments were run on a 48-core AMD EPYC 7643 CPU with hyperthreading disabled and 256GB of RAM. This corresponds to the 96-core queues of a cluster, which consist of 2-socket machines with one of these CPUs in each socket. We restrict all runs to a single socket.

\begin{table*}[ht]
    \centering
    \small
    \caption{\small Dynamic graph input data used in the experiments. We use the abbreviations listed to refer to the standard streams, and append ``-FF'' to the abbreviation to indicate the fixed forest stream.}
    \begin{tabular}
    {|c|c c|c c c|c c|c c|}
        \hline
         & & & \multicolumn{3}{c|}{Static Graph} & \multicolumn{2}{c|}{Standard} & \multicolumn{2}{c|}{Fixed Forest} \\
        Category & Dataset & Abbrv. & $|\nodesize|$ & $|\edgesize|$ & $2|\edgesize|/|\nodesize|$ & $|U|$ & $|Q|$ & $|U|$ & $|Q|$ \\
        \hline
        \multirow{5}{*}{Dense}
         & Kronecker-13 & K13 & $2^{13}$ & 17M  & 4.2K & 18M  & 1.9M & 33M  & 3.7M \\
         & Kronecker-15 & K15 & $2^{15}$ & 270M & 17K  & 280M & 31M  & 530M & 59M  \\
         & Kronecker-16 & K16 & $2^{16}$ & 1.1B & 34K  & 1.1B & 120M & 2.1B & 240M \\
         & Kronecker-17 & K17 & $2^{17}$ & 4.3B & 66K  & 4.5B & 500M & 8.6B & 950M \\
         & Kronecker-18 & K18 & $2^{18}$ & 17B  & 130K & 18B  & 2.0B & 38B  & 3.8B \\
        \hline
        \multirow{4}{*}{Random}
         & Random-$N$          & RN   & 300K & 600K & 2  & 0.60M & 67K  & 13M & 1.4M \\
         & Random-$N\lg N$     & RLOG & 100K & 1.6M & 16 & 1.60M & 180K & 60M & 6.7M \\
         & Random-$N\sqrt{N}$  & RSQR & 20K  & 1.6M & 80 & 1.60M & 180K & 63M & 7.0M \\
         & Random-10-Comp & R10C & 100K & 1.6M & 16 & 1.60M & 180K & 60M & 6.7M \\
        \hline
        \multirow{5}{*}{Real-World}
         & Email-DNC        & DNC   & 1.9K & 4.4K & 2.4 & 8.8K & 970  & 103K  & 11K  \\
         & Tech-As-Topology & TECH  & 35K  & 110K & 3.1 & 220K & 24K  & 2.95M & 330K \\
         & Enron            & ENRN & 87K  & 300K & 3.4 & 590K & 66K  & 8.53M & 950K \\
         & Twitter          & TWIT  & 81K  & 580K & 7.2 & 1.8M & 200K & 21.1M & 2.3M \\
         & Stanford         & STAN  & 280K & 1.7M & 5.9 & 2.3M & 260K & 57.0M & 6.3M \\
        \hline
    \end{tabular}
    \label{tab:datasets}
\end{table*}

Table~\ref{tab:datasets} lists all the datasets we used for our experiments. These are a mix of synthetic Kronecker datasets from Tench \etal~\cite{tench2024graphzeppelin}, random graphs from Federov \etal~\cite{fedorov2021scalable}, and real-world graphs from Federov \etal and Chen \etal~\cite{chen2022dynamic}.
For each graph dataset, we generate a \defn{standard stream} of updates by inserting all edges and then deleting all edges.
For each dataset, we additionally generate a \defn{fixed forest stream} which first inserts a fixed spanning forest of the graph, and then repeatedly inserts and deletes all of the other edges not in the fixed spanning forest.
We generate both streams, because the standard streams have a high proportion of \treeedge updates and the fixed forest streams have a higher proportion of \nontreeedge updates.
We insert random connectivity queries into the streams in several small bursts of consecutive queries. Approximately 10\% of stream operations are queries.
Appendix~\ref{app:datasets} contains the complete details about how we generate streams of updates and queries for our datasets.

We compare our dynamic connectivity system with three other recently published works that tackle the same problem. These systems are \dtree~\cite{chen2022dynamic}, \idtree~\cite{xu2024constant}, and \cdcname~\cite{fedorov2021scalable}.
\dtree, \idtree, and \cdcname all maintain a lossless representation of all edges in the graph and use different strategies to search these edges for a replacement whenever a spanning forest deletion occurs.
We also compared against the structural trees implementation in the experimental survey of Chen~\etal~\cite{chen2025experimental}, but found that \dtree performed strictly better on our datasets, thus we omit structural trees from our experimental evaluation.
To compare against another sketch-based dynamic connectivity baseline, we also developed an implementation of \gibb. Prior to this work no implementations of a sketch-based dynamic connectivity algorithms existed, to the best of our knowledge.
Appendix~\ref{app:baselines} contains summaries of all the baselines systems and how we configure them in our experiments.

\section{Results}\label{sec:results}

\subsection{\sysname vs \Gibb \locked}\label{sec:gibb_results}
Figure~\ref{fig:gibb_plot} shows a comparison of the update and query throughput of \sysname versus an implementation of \gibb on a few representative inputs. We test \sysname with and without the reduced height skip list optimization (see Section~\ref{sec:implementation_details}).
\sysname performs updates significantly faster than \gibb on all inputs and is as much as two orders of magnitude faster than \gibb on the dense inputs. We also see that the reduced height optimization gives a notable improvement in update speed.
For query performance, \sysname is also significantly faster than \gibb. The reduced height optimization is highly impactful for query speed, yielding close to an order of magnitude speedup in some cases.

\begin{figure}[t]
    \centering
    \includegraphics[width=0.95\linewidth]{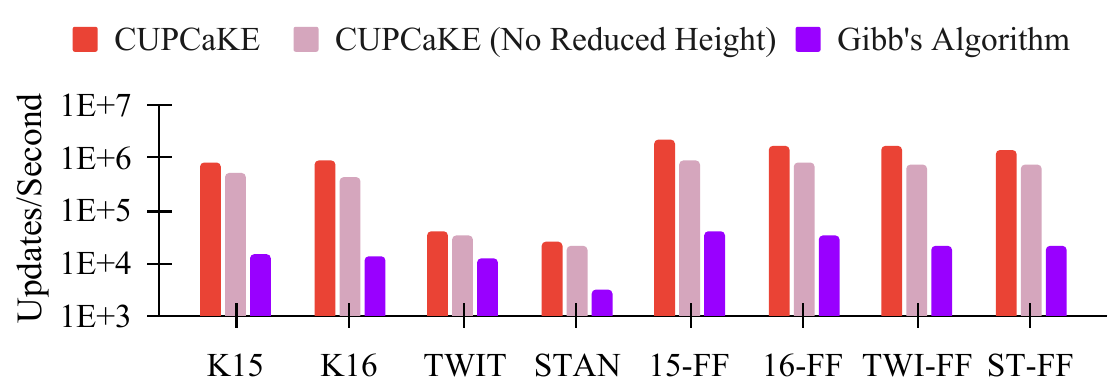}
    \includegraphics[width=0.95\linewidth]{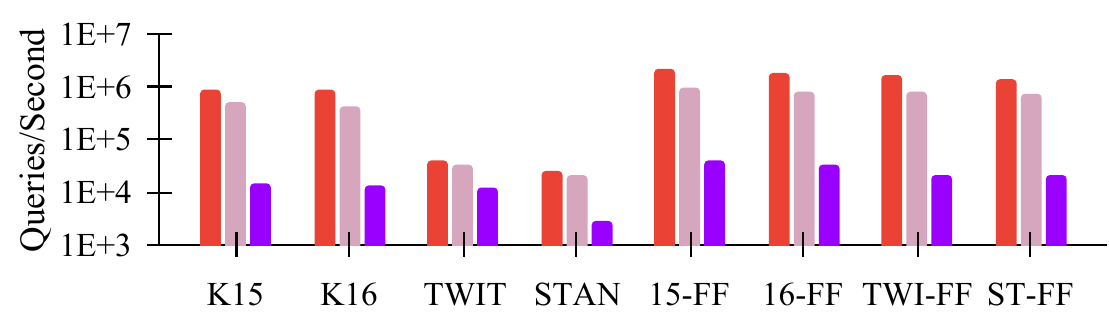}
    \caption{\small
    \sysname's updates and queries are 1-2 orders of magnitude faster than \gibb. 
    }
    \label{fig:gibb_plot}
\end{figure}

\begin{table}[ht]
\centering
\small
\caption{\small Memory usage for \sysname, \dtree, \idtree, and \cdcname. ``DNF'' indicates that the system ran out of memory. The best results are bolded for each dataset.}
\begin{filecontents*}{figures/big-spreadsheet.tsv}
	Updates/second			Queries/second			Operations/second	Total Time (ms)			Peak Memory (GB)			
	Our System	D-Tree	ID-Tree	Our System	D-Tree	ID-Tree	CDC	Our System	D-Tree	CDC	Our System	D-Tree	CDC	ID-Tree
K13	1.78818	451293.7222	13186716.58	3.36847	1006282.461	48390383.05	817242.9829	10387226203	40802.87818	23843.85478	1.5786	2.10336	13.87181633	0.36328125
K15	0.843606	333239.2664	15809384	1.72245	862205.5863	70442378.13	127973.7422	349954289549	876303.7465	2430627.351	3.23757	33.138	129.7344486	5.6
K16	0.879311		13793397.48	1.86704		62959406.92		1339531006844		3541059.443	5.6566			22.4
K17	0.60664		7816448.935	1.35022		34373002.07		7743850305602		3491212.114	10.6123			89.1
K18	0.599956			1.26606				31388203507530		3527441.088	21.6878			
DNC	0.0407358	67431.78497	3680468.423	0.159365	974489.3708	26944444.44	153159.0285	221229111	130.9637753	63.554856	0.00315189	0.0122023	3.746587476	0.004
TECH	0.0179663	78951.07595	3509667.539	0.0534115	643727.6664	24266936.3	73172.73157	12438321732	2765.865629	3270.535278	5.28724	0.0457182	1.307446972	0.0138
ENRN	0.0197168	76690.359	3626997.154	0.0674059	668726.7442	19186046.51	116015.1217	31085284269	7839.386731	5684.077991	12.4996	0.0995626	1.802006797	0.0314
TWIT	0.0388836	288225.0978	5906597.013	0.142161	1477447.641	21616947.69	583690.2635	46850708596	6267.171856	3364.911363	6.75533	0.136638	1.871976568	0.0438
STAN	0.0256101	232812.1506	5913143.633	0.0861779	711224.2756	14165077.64	161306.4045	93269743105	10293.17653	15924.61879	25.5175	0.384986	3.907901373	0.1
RN	0.0172922	29373.7664	1394236.69	0.0612389	143762.145	10097965.34	297678.8872	35783293146	20888.81326	2238.919281	27.1007	0.251849	4.684784344	0.0656
RLOG	0.119052	168863.7436	4180143.274	0.374792	470190.6336	29925187.03	1293964.184	13912618513	9852.21643	1373.545746	8.17272	0.341206	3.938457942	0.0601
RSQR	0.550947	358511.5014	8615629.829	1.51692	822276.8431	40871934.6	2247791.149	3020971510	4678.516863	790.686893	2.38256	0.206406	2.026728091	0.05
R10C	0.146041	258601.1144	6816371.22	0.448303	918192.2425	23088763.47	1820425.518	11351303301	6380.224172	976.306354	8.07689	0.284729	4.403608135	0.0557
K13-FF	4.56375	731517.462	26219817.41	6.92093	1046804.043	62914470.33	1155324.785	7848888662	49166.43094	32095.0857	1.58039	2.01579	15.74843062	0.63
K15-FF	2.18506	563591.0308	27730978.12	3.74838	782856.7447	57465666.7	779543.4328	260877656910	1025852.328	762985.1358	3.25024	32.2393	105.7855433	9.6
K16-FF	1.74558		25922501.6	2.98595		38227943.67		1307304958885		2435009.118	5.66563			38
K17-FF	1.62518		15543638.77	2.80655		27843478.62		5617218875883		2660896.561	10.6527			151.5927734
K18-FF	1.07058			2.0647				37462638854681		28945734.91	21.3407			
DNC-FF	1.14687	429262.9314	14949201.74	4.34441	513077.1195	38327526.13	718615.9831	92624518.17	262.7295222	159.5414	0.00315094	0.0130835	1.231875798	0.005
TECH-FF	0.810028	492049.2968	9983721.458	1.79109	823459.5117	33109260.56	1531403.232	3828462720	6399.169956	2142.096824	3.37799	0.0515165	2.115669011	0.0621
ENRN-FF	0.945956	702197.3624	13216896.55	1.85812	526435.0901	19490777.78	2692512.186	9524042169	13940.20759	3518.369592	7.22287	0.10517	2.639968765	0.188
TWIT-FF	1.66901	645322.8712	13740199.53	2.56006	877412.7753	17216342.05	5199820.923	13524390129	35280.36946	4497.134295	6.60959	0.131164	6.87162344	0.46
STAN-FF	1.37948	149544.5907	14375860.66	2.19788	476697.3665	19576038.93	1114453.075	44194650327	394415.1234	56815.03277	24.3224	0.384103	12.70154646	1.12
RN-FF	0.366682	321072.4926	7187999.359	0.851183	327560.8637	22445609.48	2994187.377	35766874771	43221.27781	4643.877703	27.677	0.241763	4.138032359	0.329
RLOG-FF	1.08213	527441.2415	17895236.17	1.931	586660.6195	35384094.34	10974066.09	58987632093	125298.7799	6083.435569	8.17558	0.334712	11.45839906	1.5
RSQR-FF	2.27727	639997.8389	22139077.69	3.33347	812944.4567	62400461.95	10554067.87	29862950451	107399.3255	6653.902162	2.38228	0.313128	11.2949775	1.25
R10C-FF	2.27798	671623.7198	20982102.27	3.00453	485902.5483	24934314.83	12248031.03	28599839959	103191.6153	5450.701982	8.17334	0.313704	11.5689965	1.49
\end{filecontents*}
\pgfplotstabletypeset[col sep=tab,
    header=false,
    zerofill, fixed, precision=3,
    every head row/.append style={after row={%
        \hline
        ~ & \multicolumn{4}{c|}{Peak Memory (GB)}\\ \hline
        \multicolumn{1}{|c|}{Dataset} & \sysacronym & \dtree & \cdcacronym & \idtree \\ \hline
        },
        output empty row
    },
    skip first n=2,
    skip rows between index={14}{28},
    columns={0,11,12,13,14},
    columns/0/.append style={string type},
	every first column/.append style={column type/.add={|}{|}},
	every last column/.append style={column type/.add={}{|}},
    every last row/.append style={after row=\hline},
    every row no 4/.append style={after row=\hline},
    every row no 9/.append style={after row=\hline},
    every row no 13/.append style={after row=\hline},
    every row no 18/.append style={after row=\hline},
    every row no 23/.append style={after row=\hline},
    every row 2 column 2/.append style={postproc cell content/.style={/pgfplots/table/@cell content/.add={}{DNF}}},
    every row 3 column 2/.append style={postproc cell content/.style={/pgfplots/table/@cell content/.add={}{DNF}}},
    every row 4 column 2/.append style={postproc cell content/.style={/pgfplots/table/@cell content/.add={}{DNF}}},
    every row 2 column 3/.append style={postproc cell content/.style={/pgfplots/table/@cell content/.add={}{DNF}}},
    every row 3 column 3/.append style={postproc cell content/.style={/pgfplots/table/@cell content/.add={}{DNF}}},
    every row 4 column 3/.append style={postproc cell content/.style={/pgfplots/table/@cell content/.add={}{DNF}}},
    every row 4 column 4/.append style={postproc cell content/.style={/pgfplots/table/@cell content/.add={}{DNF}}},
    every row 0 column 4/.append style={postproc cell content/.style={/pgfplots/table/@cell content/.add={$\bf}{$}}},
    every row 1 column 1/.append style={postproc cell content/.style={/pgfplots/table/@cell content/.add={$\bf}{$}}},
    every row 2 column 1/.append style={postproc cell content/.style={/pgfplots/table/@cell content/.add={$\bf}{$}}},
    every row 3 column 1/.append style={postproc cell content/.style={/pgfplots/table/@cell content/.add={$\bf}{$}}},
    every row 4 column 1/.append style={postproc cell content/.style={/pgfplots/table/@cell content/.add={$\bf}{$}}},
    every row 5 column 1/.append style={postproc cell content/.style={/pgfplots/table/@cell content/.add={$\bf}{$}}},
    every row 6 column 4/.append style={postproc cell content/.style={/pgfplots/table/@cell content/.add={$\bf}{$}}},
    every row 7 column 4/.append style={postproc cell content/.style={/pgfplots/table/@cell content/.add={$\bf}{$}}},    
    every row 8 column 4/.append style={postproc cell content/.style={/pgfplots/table/@cell content/.add={$\bf}{$}}},
    every row 9 column 4/.append style={postproc cell content/.style={/pgfplots/table/@cell content/.add={$\bf}{$}}},
    every row 10 column 4/.append style={postproc cell content/.style={/pgfplots/table/@cell content/.add={$\bf}{$}}},
    every row 11 column 4/.append style={postproc cell content/.style={/pgfplots/table/@cell content/.add={$\bf}{$}}},
    every row 12 column 4/.append style={postproc cell content/.style={/pgfplots/table/@cell content/.add={$\bf}{$}}},
    every row 13 column 4/.append style={postproc cell content/.style={/pgfplots/table/@cell content/.add={$\bf}{$}}},
    every row 14 column 4/.append style={postproc cell content/.style={/pgfplots/table/@cell content/.add={$\bf}{$}}},
    every row 15 column 1/.append style={postproc cell content/.style={/pgfplots/table/@cell content/.add={$\bf}{$}}},
    every row 16 column 1/.append style={postproc cell content/.style={/pgfplots/table/@cell content/.add={$\bf}{$}}},
    every row 17 column 1/.append style={postproc cell content/.style={/pgfplots/table/@cell content/.add={$\bf}{$}}},
    every row 18 column 1/.append style={postproc cell content/.style={/pgfplots/table/@cell content/.add={$\bf}{$}}},
    every row 19 column 1/.append style={postproc cell content/.style={/pgfplots/table/@cell content/.add={$\bf}{$}}},
    every row 20 column 2/.append style={postproc cell content/.style={/pgfplots/table/@cell content/.add={$\bf}{$}}},
    every row 21 column 2/.append style={postproc cell content/.style={/pgfplots/table/@cell content/.add={$\bf}{$}}},
    every row 22 column 2/.append style={postproc cell content/.style={/pgfplots/table/@cell content/.add={$\bf}{$}}},
    every row 23 column 2/.append style={postproc cell content/.style={/pgfplots/table/@cell content/.add={$\bf}{$}}},
    every row 24 column 2/.append style={postproc cell content/.style={/pgfplots/table/@cell content/.add={$\bf}{$}}},
    every row 25 column 2/.append style={postproc cell content/.style={/pgfplots/table/@cell content/.add={$\bf}{$}}},
    every row 26 column 2/.append style={postproc cell content/.style={/pgfplots/table/@cell content/.add={$\bf}{$}}},
    every row 27 column 2/.append style={postproc cell content/.style={/pgfplots/table/@cell content/.add={$\bf}{$}}},
    ]{figures/big-spreadsheet.tsv}
    \label{tab:results_space}
\end{table}

\begin{table*}[ht]
\centering
\small
\caption{\small Update and query speed for \sysname, \dtree, \idtree, and \cdcname. ``DNF'' indicates that the system ran out of memory. The best results are bolded for each dataset.}
\pgfplotstabletypeset[col sep=tab,
    header=false,
    sci,sci 10e, sci zerofill, precision=2,
    every head row/.append style={after row={%
        \hline
        ~ & \multicolumn{3}{c|}{Updates/second} & \multicolumn{3}{c|}{Queries/second} & Ops/second \\ \hline
        \multicolumn{1}{|c|}{Dataset} & \sysacronym & \dtree & \idtree & \sysacronym & \dtree & \idtree & \cdcacronym \\ \hline
        },
        output empty row
    },
    skip first n=2,
    columns={0,1,2,3,4,5,6,7},
    columns/0/.append style={string type},
    columns/1/.append style={preproc/expr={10^6*##1}},
    columns/4/.append style={preproc/expr={10^6*##1}},
	every first column/.append style={column type/.add={|}{|}},
	every last column/.append style={column type/.add={}{|}},
    every last row/.append style={after row=\hline},
    every row no 4/.append style={after row=\hline},
    every row no 9/.append style={after row=\hline},
    every row no 13/.append style={after row=\hline},
    every row no 18/.append style={after row=\hline},
    every row no 23/.append style={after row=\hline},
    every col no 3/.append style={column type/.add={}{|}},
    every col no 6/.append style={column type/.add={}{|}},
    every row 2 column 2/.append style={postproc cell content/.style={/pgfplots/table/@cell content/.add={}{DNF}}},
    every row 3 column 2/.append style={postproc cell content/.style={/pgfplots/table/@cell content/.add={}{DNF}}},
    every row 4 column 2/.append style={postproc cell content/.style={/pgfplots/table/@cell content/.add={}{DNF}}},
    every row 4 column 3/.append style={postproc cell content/.style={/pgfplots/table/@cell content/.add={}{DNF}}},
    every row 2 column 5/.append style={postproc cell content/.style={/pgfplots/table/@cell content/.add={}{DNF}}},
    every row 3 column 5/.append style={postproc cell content/.style={/pgfplots/table/@cell content/.add={}{DNF}}},
    every row 4 column 5/.append style={postproc cell content/.style={/pgfplots/table/@cell content/.add={}{DNF}}},
    every row 4 column 6/.append style={postproc cell content/.style={/pgfplots/table/@cell content/.add={}{DNF}}},
    every row 2 column 7/.append style={postproc cell content/.style={/pgfplots/table/@cell content/.add={}{DNF}}},
    every row 3 column 7/.append style={postproc cell content/.style={/pgfplots/table/@cell content/.add={}{DNF}}},
    every row 4 column 7/.append style={postproc cell content/.style={/pgfplots/table/@cell content/.add={}{DNF}}},
    every row 16 column 2/.append style={postproc cell content/.style={/pgfplots/table/@cell content/.add={}{DNF}}},
    every row 17 column 2/.append style={postproc cell content/.style={/pgfplots/table/@cell content/.add={}{DNF}}},
    every row 18 column 2/.append style={postproc cell content/.style={/pgfplots/table/@cell content/.add={}{DNF}}},
    every row 18 column 3/.append style={postproc cell content/.style={/pgfplots/table/@cell content/.add={}{DNF}}},
    every row 16 column 5/.append style={postproc cell content/.style={/pgfplots/table/@cell content/.add={}{DNF}}},
    every row 17 column 5/.append style={postproc cell content/.style={/pgfplots/table/@cell content/.add={}{DNF}}},
    every row 18 column 5/.append style={postproc cell content/.style={/pgfplots/table/@cell content/.add={}{DNF}}},
    every row 18 column 6/.append style={postproc cell content/.style={/pgfplots/table/@cell content/.add={}{DNF}}},
    every row 16 column 7/.append style={postproc cell content/.style={/pgfplots/table/@cell content/.add={}{DNF}}},
    every row 17 column 7/.append style={postproc cell content/.style={/pgfplots/table/@cell content/.add={}{DNF}}},
    every row 18 column 7/.append style={postproc cell content/.style={/pgfplots/table/@cell content/.add={}{DNF}}},
    every row 0 column 3/.append style={postproc cell content/.style={/pgfplots/table/@cell content/.add={$\bf}{$}}},
    every row 0 column 6/.append style={postproc cell content/.style={/pgfplots/table/@cell content/.add={$\bf}{$}}},
    every row 1 column 3/.append style={postproc cell content/.style={/pgfplots/table/@cell content/.add={$\bf}{$}}},
    every row 1 column 6/.append style={postproc cell content/.style={/pgfplots/table/@cell content/.add={$\bf}{$}}},
    every row 2 column 3/.append style={postproc cell content/.style={/pgfplots/table/@cell content/.add={$\bf}{$}}},
    every row 2 column 6/.append style={postproc cell content/.style={/pgfplots/table/@cell content/.add={$\bf}{$}}},
    every row 3 column 3/.append style={postproc cell content/.style={/pgfplots/table/@cell content/.add={$\bf}{$}}},
    every row 3 column 6/.append style={postproc cell content/.style={/pgfplots/table/@cell content/.add={$\bf}{$}}},
    every row 4 column 1/.append style={postproc cell content/.style={/pgfplots/table/@cell content/.add={$\bf}{$}}},
    every row 4 column 4/.append style={postproc cell content/.style={/pgfplots/table/@cell content/.add={$\bf}{$}}},
    every row 5 column 3/.append style={postproc cell content/.style={/pgfplots/table/@cell content/.add={$\bf}{$}}},
    every row 5 column 6/.append style={postproc cell content/.style={/pgfplots/table/@cell content/.add={$\bf}{$}}},
    every row 6 column 3/.append style={postproc cell content/.style={/pgfplots/table/@cell content/.add={$\bf}{$}}},
    every row 6 column 6/.append style={postproc cell content/.style={/pgfplots/table/@cell content/.add={$\bf}{$}}},
    every row 7 column 3/.append style={postproc cell content/.style={/pgfplots/table/@cell content/.add={$\bf}{$}}},
    every row 7 column 6/.append style={postproc cell content/.style={/pgfplots/table/@cell content/.add={$\bf}{$}}},
    every row 8 column 3/.append style={postproc cell content/.style={/pgfplots/table/@cell content/.add={$\bf}{$}}},
    every row 8 column 6/.append style={postproc cell content/.style={/pgfplots/table/@cell content/.add={$\bf}{$}}},
    every row 9 column 3/.append style={postproc cell content/.style={/pgfplots/table/@cell content/.add={$\bf}{$}}},
    every row 9 column 6/.append style={postproc cell content/.style={/pgfplots/table/@cell content/.add={$\bf}{$}}},
    every row 10 column 3/.append style={postproc cell content/.style={/pgfplots/table/@cell content/.add={$\bf}{$}}},
    every row 10 column 6/.append style={postproc cell content/.style={/pgfplots/table/@cell content/.add={$\bf}{$}}},
    every row 11 column 3/.append style={postproc cell content/.style={/pgfplots/table/@cell content/.add={$\bf}{$}}},
    every row 11 column 6/.append style={postproc cell content/.style={/pgfplots/table/@cell content/.add={$\bf}{$}}},
    every row 12 column 3/.append style={postproc cell content/.style={/pgfplots/table/@cell content/.add={$\bf}{$}}},
    every row 12 column 6/.append style={postproc cell content/.style={/pgfplots/table/@cell content/.add={$\bf}{$}}},
    every row 13 column 3/.append style={postproc cell content/.style={/pgfplots/table/@cell content/.add={$\bf}{$}}},
    every row 13 column 6/.append style={postproc cell content/.style={/pgfplots/table/@cell content/.add={$\bf}{$}}},
    every row 14 column 3/.append style={postproc cell content/.style={/pgfplots/table/@cell content/.add={$\bf}{$}}},
    every row 14 column 6/.append style={postproc cell content/.style={/pgfplots/table/@cell content/.add={$\bf}{$}}},
    every row 15 column 3/.append style={postproc cell content/.style={/pgfplots/table/@cell content/.add={$\bf}{$}}},
    every row 15 column 6/.append style={postproc cell content/.style={/pgfplots/table/@cell content/.add={$\bf}{$}}},
    every row 16 column 3/.append style={postproc cell content/.style={/pgfplots/table/@cell content/.add={$\bf}{$}}},
    every row 16 column 6/.append style={postproc cell content/.style={/pgfplots/table/@cell content/.add={$\bf}{$}}},
    every row 17 column 3/.append style={postproc cell content/.style={/pgfplots/table/@cell content/.add={$\bf}{$}}},
    every row 17 column 6/.append style={postproc cell content/.style={/pgfplots/table/@cell content/.add={$\bf}{$}}},
    every row 18 column 1/.append style={postproc cell content/.style={/pgfplots/table/@cell content/.add={$\bf}{$}}},
    every row 18 column 4/.append style={postproc cell content/.style={/pgfplots/table/@cell content/.add={$\bf}{$}}},
    every row 19 column 3/.append style={postproc cell content/.style={/pgfplots/table/@cell content/.add={$\bf}{$}}},
    every row 19 column 6/.append style={postproc cell content/.style={/pgfplots/table/@cell content/.add={$\bf}{$}}},
    every row 20 column 3/.append style={postproc cell content/.style={/pgfplots/table/@cell content/.add={$\bf}{$}}},
    every row 20 column 6/.append style={postproc cell content/.style={/pgfplots/table/@cell content/.add={$\bf}{$}}},
    every row 21 column 3/.append style={postproc cell content/.style={/pgfplots/table/@cell content/.add={$\bf}{$}}},
    every row 21 column 6/.append style={postproc cell content/.style={/pgfplots/table/@cell content/.add={$\bf}{$}}},
    every row 22 column 3/.append style={postproc cell content/.style={/pgfplots/table/@cell content/.add={$\bf}{$}}},
    every row 22 column 6/.append style={postproc cell content/.style={/pgfplots/table/@cell content/.add={$\bf}{$}}},
    every row 23 column 3/.append style={postproc cell content/.style={/pgfplots/table/@cell content/.add={$\bf}{$}}},
    every row 23 column 6/.append style={postproc cell content/.style={/pgfplots/table/@cell content/.add={$\bf}{$}}},
    every row 24 column 3/.append style={postproc cell content/.style={/pgfplots/table/@cell content/.add={$\bf}{$}}},
    every row 24 column 6/.append style={postproc cell content/.style={/pgfplots/table/@cell content/.add={$\bf}{$}}},
    every row 25 column 3/.append style={postproc cell content/.style={/pgfplots/table/@cell content/.add={$\bf}{$}}},
    every row 25 column 6/.append style={postproc cell content/.style={/pgfplots/table/@cell content/.add={$\bf}{$}}},
    every row 26 column 3/.append style={postproc cell content/.style={/pgfplots/table/@cell content/.add={$\bf}{$}}},
    every row 26 column 6/.append style={postproc cell content/.style={/pgfplots/table/@cell content/.add={$\bf}{$}}},
    every row 27 column 3/.append style={postproc cell content/.style={/pgfplots/table/@cell content/.add={$\bf}{$}}},
    every row 27 column 6/.append style={postproc cell content/.style={/pgfplots/table/@cell content/.add={$\bf}{$}}},
]{figures/big-spreadsheet.tsv}
    \label{tab:results}
\end{table*}

\subsection{\sysname is Compact on Dense Graphs}\label{sec:memory_results}
Table~\ref{tab:results_space} shows our results for peak memory usage. We only show results for the standard streams, since at the peak the graph is the same for both variants of the stream.
%
The memory usage of \sysname is a function of the number of vertices in the input graph and independent of the number of edges. Thus, \sysname is extraordinarily compact on dense graphs and predictably spacious on sparse graphs. This compactness is most noticeable for the largest dense graph we tested, Kronecker-18, where only \sysname was was able to process the full stream without running out of memory.
For the  Kronecker-15, Kronecker-16, and Kronecker-17 graphs, \sysname uses less space than \dtree, \idtree, and \cdcname. On the dense graphs that could be processed, \sysname uses up to $10$ times less memory than \dtree, up to $8.4$ times less memory than \idtree, and up to $40$ times less memory than \cdcname.
On the smallest dense graph, Kronecker-13, \sysname uses less space than \dtree and \cdcname, but more space than \idtree. This indicates that an average degree $> 4000$ is necessary for \sysname to have space savings over the most space efficient lossless system.
For the reasonably large sparser graphs (average degree $\leq 100$), \sysname as expected uses more memory than the lossless systems. \idtree is consistently the most space-efficient lossless system in our experiments.
%

\subsection{\sysname Maintains Reasonable Update Speeds}
\label{sec:update_results}
\label{subsec:ingest}
In this section we analyze the update throughput performance of the four systems, shown in Table~\ref{tab:results}. We also analyze the impact of input properties such as density and frequency of \treeedge updates.
\sysname maintains reasonable update speed performance, even on dense graphs which are too large to be processed in main memory by lossless systems.

In Table~\ref{tab:results} we see that \sysname can process hundreds of thousands to millions of updates per second on the denser standard streams (the Kronecker graphs and some of the random graphs). It ingests faster than both \dtree and \cdcname on these streams, but not faster than \idtree. On sparser graph streams, \sysname's ingestion is slower but within an order of magnitude of \dtree and \cdcname. \idtree consistently ingests sparse graphs an order of magnitude faster than \dtree and \cdcname.


Next we analyze the impact of \treeedge update frequency. The fixed-forest streams are designed to have very few \treeedge (and therefore isolated) updates. 
\sysname ingests the fixed-forest streams for the Kronecker graphs about twice as fast as their corresponding standard stream.
For the fixed-forest streams of sparser graphs, \sysname benefits much more and can ingest the stream up to 65$\times$ faster than its standard version (Stanford-Web).
This indicates that the worse performance of \sysname on sparse graphs than on dense graphs is not actually due to the sparsity of the graph, but rather the high frequency of \treeedge updates in the dynamic streams generated from sparse graphs.


%
%

The \dtree, \idtree, and \cdcname systems all also perform better on the fixed-forest streams than on their standard counterpart for each graph.
\sysname is faster than \dtree on all fixed-forest inputs, including the sparse graphs.
\sysname is sometimes better than and sometimes worse than \cdcname on these inputs.
Just like on the standard streams, \idtree is consistently the fastest on fixed-forest streams.
The main difference for these inputs is the improved performance of \sysname on sparse graphs.

%
For dense graphs \sysname uses a tenth of the space of lossless systems, supports microsecond-latency queries (as we show in the next subsection) and, unlike \gibb, has high update throughput. In fact, on dense graphs it processes updates 2-8 times faster than all but one lossless system. \idtree, the only system faster than \sysname on dense graphs, is about ten times faster and uses ten times the space. \sysname's update performance is a result of its parallel algorithm's asymptotically low span and work. 


\subsection{\sysname Answers Queries Quickly}\label{sec:query_results}
%

Table~\ref{tab:results} shows that \idtree consistently has the highest query throughput, often in the tens of millions of queries per second. On all dense streams and fixed-forest sparse streams, \sysname's query throughput is lower than that of \idtree but higher than \dtree. In particular, on the Kronecker datasets \sysname performs millions of queries per second, while \dtree is slower by roughly a factor of two. On sparse standard streams, \dtree's query throughput is higher than that of \sysname by roughly an order of magnitude.




The reason for the worse performance of \sysname on the standard streams of sparse graphs is because of the \speculative buffering used by \sysname as described in Section~\ref{sec:update_buffer}--\sysname must flush the ingestion buffer to accurately answer queries.
In the fixed forest streams, \treeedge updates are very infrequent and flushing the ingestion buffer of updates is not as costly. In Appendix~\ref{app:update_buffer_experiment} we carefully study the effect of changing the update buffer size in \sysname.

The inconsistent query performance of \dtree can be explained by the fact that query speed is a function of the diameter of the spanning tree maintained.
\idtree consistently performs the best because they maintain a union-find data structure over the graph which allows for answering queries in $O(\alpha(n))$ amortized time, the inverse Ackermann function.

\section{Conclusion \locked}
In this paper we show that it is possible to ``have our cake and eat it too'' by building \sysname, a dynamic connectivity implementation that uses small space on dense graphs, answers queries with microsecond latency, and processes up to millions of updates per second.
%
%

One exciting avenue for future work is a \emph{hybrid} cutset data structure which stores vertex neighborhoods in sketch or adjacency list form, depending on the size of the neighborhood. This data structure would sketch dense graph regions, represent sparse regions as lists, and dynamically switch between the two as the graph changes.  Such a data structure could alleviate the main limitations of \sysname, which is that it sketches the entire graph regardless of density (and as a result can waste space if the graph is too sparse).

\clearpage



\bibliographystyle{abbrv}
\bibliography{main}

\begin{thebibliography}{10}

\bibitem{acar2005experimental}
U.~Acar, G.~Blelloch, and J.~Vittes.
\newblock An experimental analysis of change propagation in dynamic trees.
\newblock In {\em Algorithm Engineering and Experiments (ALENEX)}, pages 41--54, 01 2005.

\bibitem{acar2019parallel}
U.~A. Acar, D.~Anderson, G.~E. Blelloch, and L.~Dhulipala.
\newblock Parallel batch-dynamic graph connectivity.
\newblock In {\em The 31st ACM Symposium on Parallelism in Algorithms and Architectures}, SPAA '19, page 381–392, New York, NY, USA, 2019. Association for Computing Machinery.

\bibitem{acar2004dynamizing}
U.~A. Acar, G.~E. Blelloch, R.~Harper, J.~L. Vittes, and S.~L.~M. Woo.
\newblock Dynamizing static algorithms, with applications to dynamic trees and history independence.
\newblock In {\em Proceedings of the Fifteenth Annual ACM-SIAM Symposium on Discrete Algorithms}, SODA '04, page 531–540, USA, 2004. Society for Industrial and Applied Mathematics.

\bibitem{clustering}
K.~J. Ahn, G.~Cormode, S.~Guha, A.~McGregor, and A.~Wirth.
\newblock Correlation clustering in data streams.
\newblock In {\em Proceedings of the 32nd International Conference on International Conference on Machine Learning - Volume 37}, ICML'15, pages 2237--2246. JMLR.org, 2015.

\bibitem{Ahn2012}
K.~J. Ahn, S.~Guha, and A.~McGregor.
\newblock Analyzing graph structure via linear measurements.
\newblock In {\em Proceedings of the twenty-third annual ACM-SIAM symposium on Discrete Algorithms}, pages 459--467. SIAM, 2012.

\bibitem{AhnGM13}
K.~J. Ahn, S.~Guha, and A.~McGregor.
\newblock Spectral sparsification in dynamic graph streams.
\newblock In {\em International Workshop on Approximation Algorithms for Combinatorial Optimization (APPROX)}, volume 8096, pages 1--10, 2013.

\bibitem{azad2018hipmcl}
A.~Azad, G.~A. Pavlopoulos, C.~A. Ouzounis, N.~C. Kyrpides, and A.~Bulu{\c{c}}.
\newblock Hipmcl: a high-performance parallel implementation of the markov clustering algorithm for large-scale networks.
\newblock {\em Nucleic acids research}, 46(6):e33--e33, 2018.

\bibitem{bender2025case}
M.~A. Bender, M.~Farach-Colton, R.~Jacob, H.~Koml{\'o}s, D.~Tench, and E.~T. West.
\newblock The case for external graph sketching.
\newblock In {\em 2025 Proceedings of the Conference on Applied and Computational Discrete Algorithms (ACDA)}, pages 115--129. SIAM, 2025.

\bibitem{dynamic_social}
T.~Y. Berger-Wolf and J.~Saia.
\newblock A framework for analysis of dynamic social networks.
\newblock In {\em Proceedings of the 12th ACM SIGKDD International Conference on Knowledge Discovery and Data Mining}, KDD '06, pages 523--528, New York, NY, USA, 2006. Association for Computing Machinery.

\bibitem{binaryforking}
G.~E. Blelloch, J.~T. Fineman, Y.~Gu, and Y.~Sun.
\newblock Optimal parallel algorithms in the binary-forking model.
\newblock In {\em Proceedings of the 32nd ACM Symposium on Parallelism in Algorithms and Architectures}, SPAA '20, page 89–102, New York, NY, USA, 2020. Association for Computing Machinery.

\bibitem{dynamic_web}
I.~Bordino and D.~Donato.
\newblock Dynamic characterization of a large web graph.

\bibitem{chen2025experimental}
Q.~Chen, M.~H. B\"{o}hlen, and S.~Helmer.
\newblock An experimental comparison of tree-data structures for connectivity queries on fully-dynamic undirected graphs.
\newblock {\em Proc. ACM Manag. Data}, 3(1), Feb. 2025.

\bibitem{chen2022dynamic}
Q.~Chen, O.~Lachish, S.~Helmer, and M.~H. B{\"o}hlen.
\newblock Dynamic spanning trees for connectivity queries on fully-dynamic undirected graphs.
\newblock {\em Proceedings of the VLDB Endowment}, 15(11):3263--3276, 2022.

\bibitem{facebookdense}
A.~Ching, S.~Edunov, M.~Kabiljo, D.~Logothetis, and S.~Muthukrishnan.
\newblock One trillion edges: Graph processing at {Facebook}-scale.
\newblock {\em Proc. VLDB Endow.}, 8(12):1804--1815, 2015.

\bibitem{ChitnisCEHMMV16}
R.~Chitnis, G.~Cormode, H.~Esfandiari, M.~Hajiaghayi, A.~McGregor, M.~Monemizadeh, and S.~Vorotnikova.
\newblock Kernelization via sampling with applications to finding matchings and related problems in dynamic graph streams.
\newblock In {\em {SODA}}, pages 1326--1344. {SIAM}, 2016.

\bibitem{cormode2014unifying}
G.~Cormode and D.~Firmani.
\newblock A unifying framework for l0-sampling algorithms.
\newblock {\em Distrib. Parallel Databases}, 32(3):315–335, sep 2014.

\bibitem{deman2025towards}
Q.~De~Man, L.~Dhulipala, A.~Karczmarz, J.~\L{}\k{a}cki, J.~Shun, and Z.~Wang.
\newblock Towards scalable and practical batch-dynamic connectivity.
\newblock {\em Proc. VLDB Endow.}, 18(3):889–901, Nov. 2024.

\bibitem{doraiswamy2009efficient}
H.~Doraiswamy and V.~Natarajan.
\newblock Efficient algorithms for computing reeb graphs.
\newblock {\em Computational Geometry}, 42(6-7):606--616, 2009.

\bibitem{eyal2005improved}
E.~Eyal and D.~Halperin.
\newblock Improved maintenance of molecular surfaces using dynamic graph connectivity.
\newblock In {\em International Workshop on Algorithms in Bioinformatics}, pages 401--413. Springer, 2005.

\bibitem{fedorov2021scalable}
A.~Fedorov, N.~Koval, and D.~Alistarh.
\newblock A scalable concurrent algorithm for dynamic connectivity.
\newblock In {\em Proceedings of the 33rd ACM Symposium on Parallelism in Algorithms and Architectures}, pages 208--220, 2021.

\bibitem{semistreaming1}
J.~Feigenbaum, S.~Kannan, A.~McGregor, S.~Suri, and J.~Zhang.
\newblock On graph problems in a semi-streaming model.
\newblock {\em Theor. Comput. Sci.}, 348(2):207--216, Dec. 2005.

\bibitem{frederickson1985data}
G.~N. Frederickson.
\newblock Data structures for on-line updating of minimum spanning trees, with applications.
\newblock {\em SIAM Journal on Computing}, 14(4):781--798, 1985.

\bibitem{frederickson1997ambivalent}
G.~N. Frederickson.
\newblock Ambivalent data structures for dynamic 2-edge-connectivity and k smallest spanning trees.
\newblock {\em SIAM Journal on Computing}, 26(2):484--538, 1997.

\bibitem{frederickson1997data}
G.~N. Frederickson.
\newblock A data structure for dynamically maintaining rooted trees.
\newblock {\em Journal of Algorithms}, 24(1):37--65, 1997.

\bibitem{GeorganasEHGATBOY2018}
E.~Georganas, R.~Egan, S.~Hofmeyr, E.~Goltsman, B.~Arndt, A.~Tritt, A.~Buluç, L.~Oliker, and K.~Yelick.
\newblock Extreme scale de novo metagenome assembly.
\newblock In {\em SC18: International Conference for High Performance Computing, Networking, Storage and Analysis}, pages 122--134, 2018.

\bibitem{gibb2015dynamic}
D.~Gibb, B.~Kapron, V.~King, and N.~Thorn.
\newblock Dynamic graph connectivity with improved worst case update time and sublinear space, 2015.

\bibitem{GuhaMT15}
S.~Guha, A.~McGregor, and D.~Tench.
\newblock Vertex and hyperedge connectivity in dynamic graph streams.
\newblock In {\em Proceedings of the 34th Annual ACM Symposium on Principles of Database Systems (PODS)}, pages 241--247. {ACM}, 2015.

\bibitem{henzinger1995randomized}
M.~R. Henzinger and V.~King.
\newblock Randomized dynamic graph algorithms with polylogarithmic time per operation.
\newblock In {\em {ACM} Symposium on Theory of Computing (STOC)}. ACM, 1995.

\bibitem{conn_apps}
M.~R. Henzinger and V.~King.
\newblock Randomized fully dynamic graph algorithms with polylogarithmic time per operation.
\newblock {\em J. ACM}, 46(4):502–516, July 1999.

\bibitem{henzinger1999constructing}
M.~R. Henzinger, V.~King, and T.~Warnow.
\newblock Constructing a tree from homeomorphic subtrees, with applications to computational evolutionary biology.
\newblock {\em Algorithmica}, 24(1):1--13, 1999.

\bibitem{henzinger1997sampling}
M.~R. Henzinger and M.~Thorup.
\newblock Sampling to provide or to bound: With applications to fully dynamic graph algorithms.
\newblock {\em Random Structures \& Algorithms}, 11(4):369--379, 1997.

\bibitem{holm2001poly}
J.~Holm, K.~De~Lichtenberg, and M.~Thorup.
\newblock Poly-logarithmic deterministic fully-dynamic algorithms for connectivity, minimum spanning tree, 2-edge, and biconnectivity.
\newblock {\em Journal of the ACM (JACM)}, 48(4):723--760, 2001.

\bibitem{huang2023fully}
S.-E. Huang, D.~Huang, T.~Kopelowitz, S.~Pettie, and M.~Thorup.
\newblock Fully dynamic connectivity in $ o (\log n (\log\log n)^ 2) $ amortized expected time.
\newblock {\em TheoretiCS}, 2, 2023.

\bibitem{KapralovLMMS13}
M.~Kapralov, Y.~T. Lee, C.~Musco, C.~Musco, and A.~Sidford.
\newblock Single pass spectral sparsification in dynamic streams.
\newblock In {\em Proceedings of the 55th Annual IEEE Symposium on Foundations of Computer Science (FOCS)}, pages 561--570, 2014.

\bibitem{KapralovW14}
M.~Kapralov and D.~P. Woodruff.
\newblock Spanners and sparsifiers in dynamic streams.
\newblock In {\em Proceedings of the 2014 ACM Symposium on Principles of Distributed Computing (PODC)}, pages 272--281. {ACM}, 2014.

\bibitem{kapron2013dynamic}
B.~M. Kapron, V.~King, and B.~Mountjoy.
\newblock Dynamic graph connectivity in polylogarithmic worst case time.
\newblock In {\em {ACM-SIAM} Symposium on Discrete Algorithms (SODA)}, 2013.

\bibitem{movingobject}
M.~Korn, D.~Sanders, and J.~Pauli.
\newblock Moving object detection by connected component labeling of point cloud registration outliers on the gpu.
\newblock In {\em VISIGRAPP (6: VISAPP)}, pages 499--508, 2017.

\bibitem{McGregorTVV15}
A.~McGregor, D.~Tench, S.~Vorotnikova, and H.~T. Vu.
\newblock Densest subgraph in dynamic graph streams.
\newblock In G.~F. Italiano, G.~Pighizzini, and D.~T. Sannella, editors, {\em Proceedings of the 40th Mathematical Foundations of Computer Science (MFCS)}, pages 472--482, Berlin, Heidelberg, 2015. Springer Berlin Heidelberg.

\bibitem{McGregorVV16}
A.~McGregor, S.~Vorotnikova, and H.~T. Vu.
\newblock Better algorithms for counting triangles in data streams.
\newblock In {\em Proceedings of the 35th Annual ACM Symposium on Principles of Database Systems (PODS)}, pages 401--411. {ACM}, 2016.

\bibitem{nurk2017metaspades}
S.~Nurk, D.~Meleshko, A.~Korobeynikov, and P.~A. Pevzner.
\newblock metaspades: a new versatile metagenomic assembler.
\newblock {\em Genome research}, 27(5):824--834, 2017.

\bibitem{sleator1983data}
D.~D. Sleator and R.~{Endre Tarjan}.
\newblock A data structure for dynamic trees.
\newblock {\em Journal of Computer and System Sciences}, 26(3):362--391, 1983.

\bibitem{tench2024graphzeppelin}
D.~Tench, E.~West, V.~Zhang, M.~Bender, A.~Chowdhury, D.~Delayo, J.~Dellas, M.~Farach-Colton, T.~Seip, and K.~Zhang.
\newblock Graphzeppelin : How to find connected components (even when graphs are dense, dynamic, and massive).
\newblock {\em ACM Transactions on Database Systems}, 49, 02 2024.

\bibitem{graphzeppelin}
D.~Tench, E.~West, V.~Zhang, M.~A. Bender, A.~Chowdhury, J.~A. Dellas, M.~Farach-Colton, T.~Seip, and K.~Zhang.
\newblock Graphzeppelin: Storage-friendly sketching for connected components on dynamic graph streams.
\newblock In {\em Proceedings of the 2022 International Conference on Management of Data}, SIGMOD '22, page 325–339, New York, NY, USA, 2022. Association for Computing Machinery.

\bibitem{landscape}
D.~Tench, E.~T. West, K.~Zhang, M.~A. Bender, D.~DeLayo, M.~Farach-Colton, G.~Gill, T.~Seip, and V.~Zhang.
\newblock Exploring the landscape of distributed graph sketching.
\newblock In {\em Proceedings of the Symposium on Algorithm Engineering and Experiments (ALENEX)}, pages 133--146, 2025.

\bibitem{thorup2000near}
M.~Thorup.
\newblock Near-optimal fully-dynamic graph connectivity.
\newblock In {\em Proceedings of the thirty-second annual ACM symposium on Theory of computing}, pages 343--350, 2000.

\bibitem{tseng2019batch}
T.~Tseng, L.~Dhulipala, and G.~Blelloch.
\newblock Batch-parallel euler tour trees.
\newblock In {\em Algorithm Engineering and Experiments (ALENEX)}, 2019.

\bibitem{van2000graph}
S.~Van~Dongen.
\newblock Graph clustering by flow simulation.
\newblock {\em PhD thesis, University of Utrecht}, 2000.

\bibitem{wulff2013faster}
C.~Wulff-Nilsen.
\newblock Faster deterministic fully-dynamic graph connectivity.
\newblock In {\em Proceedings of the twenty-fourth Annual ACM-SIAM Symposium on Discrete Algorithms}, pages 1757--1769. SIAM, 2013.

\bibitem{xu2024constant}
L.~Xu, D.~Wen, L.~Qin, R.~Li, Y.~Zhang, and X.~Lin.
\newblock Constant-time connectivity querying in dynamic graphs.
\newblock {\em Proc. ACM Manag. Data}, 2(6), Dec. 2024.

\bibitem{zhang2020parallel}
Y.~Zhang, A.~Azad, and A.~Bulu{\c{c}}.
\newblock Parallel algorithms for finding connected components using linear algebra.
\newblock {\em Journal of Parallel and Distributed Computing}, 144:14--27, 2020.

\end{thebibliography}

\clearpage
\appendix
\section{Prior Work Details}
\label{app:prelims}


We begin by summarizing prior theoretical and practical work in streaming connectivity (where the primary goal is to minimize space) and dynamic connectivity (where the primary goal is to minimize update and query time). We then describe the dynamic connectivity algorithm of Gibb~\etal~\cite{gibb2015dynamic}, which achieves both small space and low update/query time. 
We include significant technical detail in these summaries so we can compare and contrast these approaches against our own; we employ modified versions of some techniques used in these prior works in the design of our practical dense-optimized dynamic connectivity algorithm.

%

\subsection{Prior Work in Streaming Connectivity}\label{sec:streaming_model}
In the \defn{graph semi-streaming} model~\cite{semistreaming1} (sometimes just called the \defn{graph streaming} model), an algorithm is presented with a \defn{stream} of updates that define a graph. Each update is an edge insertion or deletion. The challenge in this model is to compute (perhaps approximately) some property of the graph given a single pass over the stream and memory sublinear in the size of the graph. The time required to process edge updates or compute the desired property of the graph after the stream is a second-order consideration. Graph streaming algorithms have been proposed for a multitude of problems~\cite{ChitnisCEHMMV16,KapralovLMMS13,GuhaMT15,AhnGM13,McGregorVV16,KapralovW14,McGregorTVV15,clustering}.

Specifically, stream $\graphstream$ defines a graph $\graph = (\nodes,\edges)$ with $\nodesize = |\nodes|$ and $\edgesize = |\edges|$.
Each update has the form $((u,v), \Delta)$ where $u,v \in \nodes, u \neq v$ and $\Delta \in \{-1,1\}$ where $1$ indicates an edge insertion and $-1$ indicates an edge deletion.
Let $\edges_i$ be the edge set defined by the first $i$ updates in $\graphstream$, i.e., those edges which have been inserted and not subsequently deleted. The stream may only insert edge $\edge$ at time $i$ if $\edge \notin \edges_{i-1}$, and may only delete edge $\edge$ at time $i$ if $\edge \in \edges_{i-1}$. 
Prior work~\cite{Ahn2012,tench2024graphzeppelin} on connected components in this model consider the following problem:

\begin{problem}[\textbf{The streaming connected components problem}]
Given an insert/delete edge stream $\graphstream$ that defines a graph $\graph = (\nodes, \edges)$, return a spanning forest of $\graph$.
\end{problem}

\subsubsection{Linear Sketching for Streaming Connectivity}\label{sec:streaming_solutions}
Ahn \etal~\cite{Ahn2012} initiated the field of graph sketching with their connected components sketch, which solves the streaming connected components problem in $O(\nodesize \log^3\nodesize)$ space. We refer to their algorithm as \defn{\agmname}. A key subproblem in their algorithm is \defn{$\ell_0$-sampling}: a vector $\myvec$ of length $n$ is defined by an input stream of \defn{updates} of the form $(i,\Delta)$ where value $\Delta$ is added to $\myvec_i$, and the task is to \defn{query} for a nonzero element of $\myvec$ using $o(n)$ space. They use an $\ell_0$-sampler (also called an $\ell_0$-sketch) due to Cormode \etal~\cite{cormode2014unifying} (Theorem~\ref{thm:cormodesketch}) with update and query times of $O(\log(n)\log(1/\delta))$.

Recent graph sketching implementation work introduces \cameosketch~\cite{landscape}, an $\ell_0$-sampler for vectors in $\mathbb{F}_2^n$ with improved worst-case update time $O(\log(1/\delta))$ (Theorem~\ref{thm:cameosketch}) which is $O(1)$ for a fixed constant $\delta$.
We also note that \cameosketch can be easily modified to also support sketch queries in $O(\log(1/\delta))$ time.

\begin{theorem}[Adapted from~\cite{cormode2014unifying}, Theorem 1]\label{thm:cormodesketch}
Given a 2-wise independent hash family $\hashfamily$ and an input vector $\myvec \in \mathbb{Z}^n$, there is an $\ell_0$-sampler using $O(\log^2(n)\log(1/\delta))$ bits of memory that succeeds with probability at least $1 - \delta$. 
\end{theorem}

\begin{theorem}[Adapted from ~\cite{landscape}, Theorem 4.2]\label{thm:cameosketch}
\cameosketch is an $\ell_0$-sampler that, for vector $x \in \mathbb{F}_2^n$, uses $O(\log^2(n)\log(1/\delta))$ bits of memory, has worst-case update time $O(\log(1/\delta))$, and succeeds with probability at least $1-\delta$.
\end{theorem}

We denote the $\ell_0$-sketch of a vector $\myvec$ as $\lsketch(\myvec)$.  The sketch of a vector is a linear function, that is for any vectors $\myvec$ and $\myothervec$, $\lsketch(\myvec) + \lsketch(\myothervec) = \lsketch(\myvec+\myothervec)$.
In this paper, we will abstract away most of the details of $\ell_0$-sketching algorithms, but note that any sketch $\lsketch(\myvec)$ is laid out as an array of $O(\log(n)\log(1/\delta))$ machine words in RAM, and to compute $\lsketch(\myvec) + \lsketch(\myothervec)$ we simply perform elementwise XOR. Therefore sketch addition takes $O(\log(n)\log(1/\delta))$ time.

\myparagraph{The \Agmname Algorithm}
The main idea of Ahn \etal's algorithm is to sketch each vertex neighborhood, represented in a specific vector format (called the characteristic vector) such that the sum of neighborhood vectors for vertices $\nodesubset\subseteq\nodes$ encodes precisely the set of edges in the cut $(\nodesubset, \nodes \setminus \nodesubset)$.
Ahn \etal define $\charvec_u \in \mathbb{F}_2^{\binom{\nodesize}{2}}$ as the \defn{characteristic vector of a vertex} $u$, where each nonzero element of the vector $\charvec_u$ denotes an edge incident to $u$.
That is for all $0 \leq v < \nodesize$, $\charvec_u[(u,v)] = 1$ if $(u, v) \in \edges$, $0$ otherwise.
%
Summing characteristic vectors for vertices in a subset $\nodesubset$ cancels out edges internal to $\nodesubset$: for $(u,v)\in\edges$, $\charvec_u[(u,v)]+\charvec_v[(u,v)]=0$.

\agmname maintains $O(\log \nodesize)$ sketches of the characteristic vector $\charvec_u$ for each vertex $u \in \nodes$.
%
When an update $((u,v),\Delta)$ is processed, the algorithm simply updates all of the $O(\log\nodesize)$ $\ell_0$-sketches for $u$ and for $v$.
Using the $\ell_0$-sketch of Cormode \etal~\cite{cormode2014unifying} gives updates in $O(\log^2(\nodesize)\log(1/\delta))$ time. Using \cameosketch~\cite{landscape} as the $\ell_0$-sketch allows for updates in $O(\log(\nodesize)\log(1/\delta))$ time.

To find the connected components of $\graph$, the algorithm runs Boruvka's algorithm for connected components of a graph. To recover edges incident to a supernode, the algorithm queries the sketch for that supernode. When merging supernodes, their sketches are merged as well.
%
%
%
Ahn \etal show that, for $\delta = 1/100$, the correct connected components will be found in $O(\log\nodesize)$ rounds w.h.p. and thus maintaining $O(\log\nodesize)$ sketches is sufficient for the algorithm to be correct w.h.p (which here and elsewhere in the paper means with probability at least $1 - 1/\nodesize^c$ for some constant $c$). 
%
Since the sketches for any vertex have total size $O(\log^3\nodesize)$ bits, the entire data structure has size $O(\nodesize\log^3\nodesize)$ bits.

\myparagraph{Limitations of \Streamingname}
A limitation of \agmname is its performance in settings such as the dynamic graph model, where frequent connectivity queries occur in between the stream updates.
In order to answer such queries, the \agmname algorithm first finds the connected components of $\graph$ in $O(\nodesize \log^2 \nodesize)$ time and then checks whether $u$ and $v$ are in the same component.


\subsection{Prior Work in Dynamic Connectivity}\label{sec:dynamic_model}
In the \defn{dynamic graph} model, the challenge is to efficiently compute a specific structural property of a changing graph.
Specifically, an algorithm is presented with a sequence of updates (each an edge insertion or deletion). Unlike the semi-streaming model, in the dynamic model the algorithm may be asked to answer queries about the desired property of the current graph at any moment within the sequence of updates.

Specifically the algorithm is given a sequence of updates $\graphstream$ which must be processed in order.
Any prefix of the first $i$ updates in $\graphstream$ defines a graph $\graph_i = (\nodes_i,\edges_i)$.
Immediately after processing the $i$-th update , one or more queries may be issued and the algorithm must compute the answers for graph $\graph_i$ before receiving the  $i+1$-th update.
The challenge in this model is to process updates and answer queries as quickly as possible. The space usage of the algorithm is a second-order consideration.
In the dynamic connectivity problem, the structural property of the graph of interest is connectivity. After processing update $i$, an algorithm must answer connectivity queries $\connected(u, v)$ on vertices $u$ and $v$ in $\nodes$, which return whether there exists a path between $u$ and $v$ in $\graph_i$.

\begin{problem}[\textbf{The fully-dynamic connectivity problem}]



Given a sequence $\graphstream$ of edge insertions or deletions which define a changing graph $\graph$, maintain a data structure which processes one update of $\graphstream$ at a time, and can answer any connectivity query on $\graph$ after each update.
\end{problem}

\subsubsection{Lossless Dynamic Connectivity Algorithms}\label{app:dynamic_solutions}

Henzinger and King~\cite{henzinger1995randomized} produced the first dynamic connectivity algorithm with poly-logarithmic (amortized) update and query time.
This inspired many other algorithms, all of which maintain a lossless representation of the graph as part of their data structure. We refer to this class of algorithms as \defn{\losslessname dynamic connectivity algorithms}.
%
All of these algorithms maintain a spanning forest of the graph, assign a level to each edge, and store each edge explicitly. When a replacement edge must be found, the algorithm may search through many edges, reducing the level of those that are not a valid replacement edge. Since each edge can only be pushed down a limited number of times, these algorithms achieve poly-logarithmic amortized update time.

\myparagraph{Limitations of \Losslessname Algorithms}
%
Explicitly storing all of the edges requires a prohibitive amount of space for dense graphs.
The space usage of \losslessname dynamic connectivity algorithms is typically $O(\nodesize \log \nodesize + \edgesize)$ or $O(\nodesize + \edgesize)$ words (where a word is usually $O(\log \nodesize)$ bits), which can be asymptotically larger than the space usage of \agmname for a sufficiently dense graph.
%
Another limitation is that although these algorithms have amortized poly-logarithmic update time guarantees, they still may have to search through a large amount of edges in the worst case to find a replacement edge, which requires traversing a component and possibly searching all of the edges coming out of it.

\subsection{Detailed Explanation of \Gibb}\label{app:gibb}
Algorithm~\ref{alg:gibb_update} shows the pseudo-code for edge insertions and deletions in \gibb~\cite{gibb2015dynamic}.
Figure~\ref{fig:gibb_update} shows an example of an edge insertion.

\begin{figure*}[ht]
    \centering
    \includegraphics[width=\linewidth]{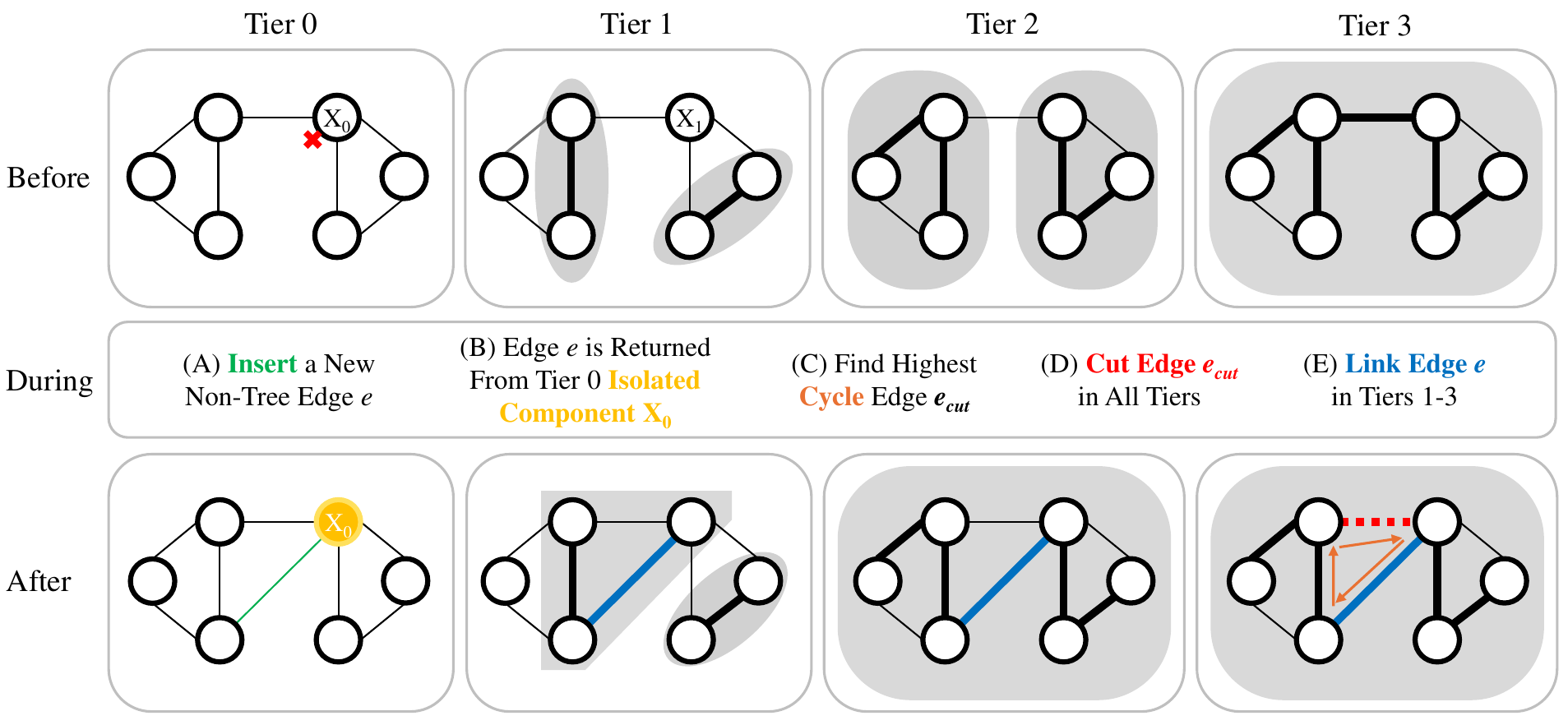}
    \caption{\small
    An example of an edge insertion in \gibb that causes an isolated component. Thick edges represent the tree edges in each tier, while thin edges represent the non-tree edges. Shaded gray regions represent components of at least two nodes. 
    Although $X_0 = X_1$ before the update, component $X_0$ is not isolated because its sketch failed to return an edge. Adding the new edge (the green edge in tier 0) causes component $X_0$ to be isolated ($X_0$ now randomly returns an edge), and the green edge is recovered from the sketch of $X_0$. 
    Making this edge a tree edge would cause a cycle at the top tier, so the lowest edge in the cycle (the dotted red line) is removed and the newly sampled edge (solid dark blue line) is added at tiers 1-3.
    }
    \label{fig:gibb_update}
\end{figure*}

\begin{algorithm}[ht]
\caption{$\mathsf{Update}(e=(u,v),\text{bool Deletion})$}
\label{alg:gibb_update}
\begin{algorithmic}[1]
    \For {$i \in [0,top]$} 
        \State $\cutset_i.\update(e)$ \label{line:gibb_update_sketch}
        \If {$\text{Deletion} \land e \in \edges_{\forest_i}$} $\cutset_i.\cut(u,v)$ \EndIf \label{line:gibb_update_cut}
    \EndFor
    \For {$i \in [0,top-1]$}\label{line:gibb_level_loop}
    \For {$w \in \{u,v\}$}
        \State $\component_i \gets$ level $i$ component containing $w$
        \State $\component_{i+1} \gets $ level $i+1$ component containing $w$
        \State $(u_{link}, v_{link}) \gets \cutset_i.\query(w)$
        \If{$((u_{link}, v_{link}) \neq \emptyset) \land (\component_i = \component_{i+1})$} \label{line:gibb_link_check}
            \If{$\tree.\connected(u_{link}, v_{link})$} \label{line:gibb_cut_start}
                \State ($u_{cut},v_{cut},\ell) \gets \tree.\pathquery(u_{link},v_{link})$
                \For {$j \in [\ell,top]$} $\cutset_j.\cut(u_{cut},v_{cut})$ \EndFor
            \EndIf \label{line:gibb_cut_end}
            \For {$j \in [i+1,top]$} $\cutset_j.\link(u_{link},v_{link})$ \EndFor \label{line:gibb_link}
        \EndIf
    \EndFor
    \EndFor
\end{algorithmic}
\end{algorithm}

First for each level $i$, $0 \leq i \leq top$, $\cutset_i.\update(e)$ is called (line~\ref{line:gibb_update_sketch}). If the update is an edge deletion and $e \in \edges_{\forest_i}$ then $\cut(u,v)$ is called on $\cutset_i$ (line~\ref{line:gibb_update_cut}).
Throughout the algorithm, whenever a $\link$ or $\cut$ occurs in $\cutset_{top}$, the same operation is done in $\tree$ (omitted from pseudo-code).

\myparagraph{Finding Isolated Components}
Next the algorithm will check for components that violate Invariant~\ref{inv:gibb3}, starting at level $0$ and going up through all the levels. Such components are referred to as \defn{isolated components}.
Let $\component_i$ indicate a component in $\forest_i$, and $\component_{i+1}$ indicate the component in $\forest_{i+1}$ that contains $\component_i$.
$\component_i$ may become isolated if (1) an edge in $\component_{i+1}$ was deleted, (2) $\component_i$ now successfully returns an edge when it previously didn't (e.g. a sketch previously failed to return an edge but now does due to a sketch update), or (3) $\component_i$ was modified (e.g. from the linking and cutting described next).
Importantly, a component can only become isolated due to an update if it contains $u$ or $v$.
Checking if a component is isolated can be done by querying $\cutset_i$ with vertices $u$ and $v$. If the query is successful (returning an edge $e_{link}=(u_{link},v_{link})$) and $\component_{i+1}$ is equal to $\component_i$ (equality can be checked by comparing the size of the components which can be obtained from the dynamic tree data structure), then $\component_i$ is isolated (line~\ref{line:gibb_link_check}).
Then $\link(u_{link},v_{link})$ should be called at all levels above $i$ to restore Invariant~\ref{inv:gibb3} while also maintaining Invariant~\ref{inv:gibb2} (line~\ref{line:gibb_link}).

\myparagraph{Removing a Cycle Edge}
It may not be possible to link the edge $\edge_{link}$ in $\forest_j$ if $u_{link}$ and $v_{link}$ are already connected in $\forest_j$, because then this would form a cycle. So it is necessary to first find the minimum level $\ell > i$ where $u_{link}$ is connected to $v_{link}$, and disconnect them in all levels $\geq \ell$ (lines~\ref{line:gibb_cut_start}--\ref{line:gibb_cut_end}).
The lowest level where two vertices are connected corresponds to the maximum weight edge on their path in $\tree$.
To find this level (and a corresponding edge to cut), the algorithm calls $\tree.\pathquery(u_{link},v_{link})$.

\myparagraph{Analysis}
The authors~\cite{gibb2015dynamic} prove that for a single update, the invariants are maintained w.h.p. Thus the answers to connectivity queries are correct w.h.p. across a polynomial number of updates.
Connectivity queries can be answered by checking if $u$ and $v$ are in the same component in $\tree$ or $\forest_{top}$ in $O(\log \nodesize)$ time.
Queries can also be improved to $O(\log \nodesize / \log\log \nodesize)$ time without changing the update complexity by maintaining a separate dynamic tree with $\Theta(\log \nodesize)$ fan-out~\cite{gibb2015dynamic}.

Now we will analyze the performance of normal and isolated updates. In \gibb, both normal and isolated updates on edge $e$ first call $\update(e)$ at each of the $O(\log \nodesize)$ levels. Due to Lemma~\ref{lem:cutset}, this part takes $O(\log^3 \nodesize)$ time.
In a normal update, the remainder of the update algorithm simply calls $\query(u)$ and $\query(v)$ on each cutset data structure over the $O(\log \nodesize)$ levels, for a total of $O(\log^2 \nodesize)$ time.
Now consider an isolated update. In the worst case, processing each level (each iteration of the main loop on line~\ref{line:gibb_level_loop} of the pseudo-code) may cause all levels above it to perform a link and/or cut of the cutset data structure.
A link or cut in a cutset takes $O(\log^2 \nodesize)$ time (Lemma~\ref{lem:cutset}), so across all $O(\log \nodesize)$ levels this results in a worst-case update time of $O(\log^4 \nodesize)$.
This yields the following lemma analyzing \gibb:
\begin{lemma}
    \gibb performs isolated updates in $O(\log^4 \nodesize)$ time and normal updates in $O(\log^3 \nodesize)$ time.
\end{lemma}

\section{Data Structures That We Use}

\subsection{Dynamic Trees}\label{app:dynamic_trees}
Much prior work in dynamic connectivity (as well as the algorithms in this paper) makes use of a class of data structure called a \defn{dynamic tree}~\cite{sleator1983data,frederickson1985data,frederickson1997ambivalent,frederickson1997data,henzinger1995randomized,acar2004dynamizing,acar2005experimental}. We summarize them here.
First introduced by Sleator and Tarjan~\cite{sleator1983data}, the dynamic trees problem is to maintain a data structure representing a forest subject to updates (edge insertions and deletions) on a fixed set of vertices.
The most common type of query that all known dynamic tree data structures can answer is whether or not two vertices in the forest are connected.
At its core, we define the following interface for the dynamic trees problem on a forest $\forest$:
\begin{itemize}[topsep=0pt,itemsep=0pt,parsep=0pt,leftmargin=15pt]
    \item $\link(u,v)$: Given two vertices $u$ and $v$ in $\forest$ that are not connected, insert the edge $(u,v)$ into $\forest$.
    \item $\cut(u,v)$: Given two vertices $u$ and $v$ in $\forest$ with an edge between them, delete the edge $(u,v)$ from $\forest$.
    \item $\connected(u,v)$: Given two vertices $u$ and $v$ in $\forest$, determine whether they are connected by a path in $\forest$.
\end{itemize}
Dynamic trees may also support more complex queries about forests with edge or vertex values in some domain $\domain$ and some commutative and associative \defn{aggregate function} $g : \mathcal{D}^2 \rightarrow \mathcal{D}$.
We formally define two additional types of queries:
\begin{itemize}[topsep=0pt,itemsep=0pt,parsep=0pt,leftmargin=15pt]
    \item $\pathquery(u,v)$: Given two connected vertices $u$ and $v$ in $\forest$, compute $g$ over the edges values on the path from $u$ to $v$.
    \item $\subtreequery(v)$: Given a vertex $v$ in $\forest$, return the compute $g$ over the vertex values in the component containing $v$.
\end{itemize}
%
In this paper, we use Euler tour trees~\cite{henzinger1995randomized} and link-cut trees~\cite{sleator1983data} because they support subtree queries and path queries respectively.
All operations are in logarithmic time.
Other examples of dynamic trees include topology trees~\cite{frederickson1985data,frederickson1997ambivalent,frederickson1997data}, and rake-compress trees~\cite{acar2004dynamizing,acar2005experimental}.

\subsection{Skip Lists}\label{app:skiplist}

Skip-lists are an ordered sequence data structure that use probabilistic balancing. It consists of several levels of linked lists where the bottom level contains every element in the data structure, and each element with a node at some level $i$ has a node in level $i+1$ independently with constant probability $p \in (0,1)$ (typically $1/2$).

\myparagraph{Reduced Height Skip List}
In a traditional skip-list the height of each element is a random variable $X$ chosen independently from a geometric distribution $X \sim \geo(1/2)$.
A reduced height skip-list uses a non-constant probability $1 / \log \nodesize$ so that the height of each element is a random variable $X'$ independently chosen from a geometric distribution $X' \sim \geo(1 - 1 / \log \nodesize)$. The following two lemmas prove bounds on the height and length of paths in the data structure. We omit the proofs due to space constraints but they follow from a basic probabilistic analysis.

\begin{lemma}
    The maximum height of an element in a reduced height skip-list with $\nodesize$ elements is $O(\log \nodesize / \log\log \nodesize)$ w.h.p.
\end{lemma}
\begin{proof}
    Consider a single element in the reduced height skip-list whose height is a random variable $X' \sim \geo(1 / \log\log \nodesize)$.
    \begin{align*}
    \prob{X' \geq c' \frac{\log \nodesize}{\log\log \nodesize}} = (1/\log \nodesize)^{c' \frac{\log \nodesize}{\log\log \nodesize}} = 1/\nodesize^{c'}
    \end{align*}
    Let $h$ be the maximum height of any element. Applying the union bound over the $\nodesize$ elements and using $c' = c+1$:
    \[ \prob{h \geq c' \log \nodesize / \log\log \nodesize} \leq \nodesize \cdot 1/\nodesize^{c+1} = 1/\nodesize^c \]
    Thus $h = O(\log \nodesize / \log\log \nodesize)$ with probability at least $1-1/\nodesize^c$.
\end{proof}

\begin{lemma}
    The number of nodes on each search path in a reduced height skip-list with $\nodesize$ elements is $O(\log^2 \nodesize)$ w.h.p.
\end{lemma}
\begin{proof}
    Consider traversing a given search path in reverse, starting from the bottom-level node. At any point the path travels up with probability at least $1 / \log \nodesize$, otherwise it travels left.
    The question is how many nodes are traversed before the top level is reached (e.g. the path travels up $c \log \nodesize / \log\log \nodesize$ times). This question can be thought of as flipping a biased coin where heads means the path travels up and tails means it travels left.
    We show that in $6c \log^2 \nodesize$ flips of a biased coin with probability $1 / \log \nodesize$ of heads, there will be at least $c \log \nodesize / \log\log \nodesize$ heads with probability at least $1-1/\nodesize^c$:
    \begin{align*}
    &\prob{\text{In } 6c \log^2 \nodesize \text{ coin flips there's }\leq c \frac{\log \nodesize}{\log\log \nodesize} \text{ heads}} \\
    &\leq \binom{6c \log^2 \nodesize}{c\log \nodesize / \log\log \nodesize}\left(1-\frac{1}{\log \nodesize}\right)^{6c \log^2 \nodesize - c \frac{\log \nodesize}{\log\log \nodesize}} \\
    &\leq \left(\frac{6ec \log^2 \nodesize}{c \log \nodesize / \log\log \nodesize}\right)^{c \frac{\log \nodesize}{\log\log \nodesize}}\left(\frac{1}{e}\right)^{6c\log \nodesize - c \frac{1}{\log\log \nodesize}} \\
    &\leq (6e\log \nodesize)^{c \frac{\log \nodesize}{\log\log \nodesize}}\left(\frac{1}{e}\right)^{2c\log \nodesize+3c \frac{\log \nodesize}{\log\log \nodesize}} \\
    &\leq \left(\frac{6e\log \nodesize}{e^3}\right)^{c \frac{\log \nodesize}{\log\log \nodesize}}\left(\frac{1}{e}\right)^{2c\log \nodesize}\leq \nodesize^c\left(\frac{1}{\nodesize^{2c}}\right) \leq \frac{1}{\nodesize^c}
    \end{align*}
    Using the union bound over the $\nodesize$ search paths, the probability that any search path is longer than $6c\log^2\nodesize$ is $1/\nodesize^{c-1}$, thus the number of nodes on each search path is $O(\log^2 \nodesize)$ w.h.p.
\end{proof}

To take advantage of the reduced height property it is necessary to maintain pointers from each skip-list node to its \defn{parent}. We define the parent of a level $i$ skip-list node $w$ as the rightmost level $i+1$ node whose corresponding level $i$ node is to the left of $w$. In other words, the parent is the last level $i+1$ node on any search path that passes through $w$.
It is easy to see that finding the root of the reduced height skip-list containing an element can be done in $O(\log \nodesize / \log\log \nodesize)$ time w.h.p. if parent pointers are maintained.
Maintaining parent pointers during splits or joins does not affect the overall cost.

Prior work~\cite{tseng2019batch} describes how to augment skip-lists with values in some domain $\domain$ and a commutative and associative function $g : \domain^2 \rightarrow \domain$.
The same techniques may be applied to augment a reduced height skip-list. As a result we get the following lemma:

\begin{lemma}
    There exists an augmented dynamic tree data structure that performs root finding in $O(\log \nodesize / \log\log \nodesize)$ time w.h.p., performs augmented value updates in $O(w \cdot \log \nodesize / \log\log \nodesize)$ time w.h.p. where $w$ is the time to update one augmented value, and performs links and cuts in $O(y \cdot \log^2 \nodesize)$ time w.h.p. where $y$ is the time for one computation of the aggregate function over augmented values.
\end{lemma}
\begin{proof}
    Updating an augmented value only requires modifying the rightmost node on each level in a single search path (see~\cite{tseng2019batch} for details), thus the cost is proportional to the height multiplied by the cost of updating a single augmented value. Since the height of a reduced height skip-list is $O(\log \nodesize / \log\log \nodesize)$ w.h.p., this takes $O(w \cdot \log \nodesize / \log\log \nodesize)$ time w.h.p. where $w$ is the time to update one augmented value.
    The cost for joins and splits is proportional to the length of a search path multiplied by the cost $y$ of a single computation of the aggregate function over augmented values (see~\cite{tseng2019batch} for details). Since the length of a search path in a reduced height skip-list is $O(\log^2 \nodesize)$ w.h.p., joins and splits take $O(y \cdot \log^2 \nodesize)$ time w.h.p.
    Maintaining parent pointers during joins and splits on the reduced height skip-list takes time proportional to the length of a search path which is $O(\log^2 \nodesize)$ w.h.p. because we can update the pointers as we traverse the path.
    Using reduced height skip-lists to implement Euler tour trees completes the proof.
\end{proof}

\section{Experimental Setup Details}\label{app:experiments}

\subsection{Machine Details and Datasets}
\label{app:datasets}

We implemented \sysname as a C++17 executable compiled with Open MPI version 5.0.2 and GCC version 13.2.0. All experiments were run on a 48-core AMD EPYC 7643 CPU with hyperthreading disabled and 256GB of RAM. This corresponds to the 96-core queues of a cluster, which consist of 2-socket machines with one of these CPUs in each socket. We use slurm to restrict all runs to a single socket. 

\myparagraph{Datasets}
Table~\ref{tab:datasets} lists all the datasets we used for our experiments. These are a mix of synthetic Kronecker datasets from Tench \etal~\cite{tench2024graphzeppelin}, random graphs from Federov \etal~\cite{fedorov2021scalable}, and real-world graphs from Federov \etal and Chen \etal~\cite{chen2022dynamic}.
These inputs represent a mix of sparse (real-world and some random) and dense (kronecker and some random) graphs, with varied structural properties such as diameter and number of components. 

\myparagraph{Converting Static Graphs to Streams}
We convert the static graph data 
(the real-world and random graph datasets) 
into a stream of updates in two ways and therefore generate two different datasets. We generate the \defn{standard stream} of updates by inserting all edges and then deleting all edges.
We generate a \defn{fixed forest stream} of updates by first computing a spanning forest of the static graph and creating edge insertion operations for each edge in the spanning forest. Then, with the remaining edges that were not part of the spanning forest, we fully insert them all and then fully delete them all in a random order 20 times (just once for already large streams).



\myparagraph{Adding Queries to Streams} We augment the sequences of edge insertions and deletions from all datasets with connectivity queries where $u$ and $v$ are chosen independently and uniformly at random from $\nodes$.  Specifically, we add a burst of queries in between periods of updates with length $\queryperiod$ selected uniformly at random between $1000$ and $2000$ updates. The number of queries in each burst is $\queryperiod/9$, so the number queries is $10\%$ of the total operations in the stream. 


\myparagraph{Measuring the Impact of Stream Properties}
The two main properties of the input streams that impact performance are density and proportion of \treeedge updates. These two are often correlated; dense graphs with random insertions and deletions will have a much higher proportion of \nontreeedge updates than sparse graphs.
We design experiments that decorrelate density and \treeedge update frequency as much as possible.
We accomplish this with the standard stream and the fixed-forest stream variants (see earlier in this section).
The standard streams have a high proportion of \treeedge updates, and the fixed forest streams have a much higher proportion of \nontreeedge updates.
%
Testing these systems on both stream types allows us to indirectly measure the performance difference between \treeedge and \nontreeedge updates.

\subsection{Baseline Systems}\label{app:baselines}
We compare our dynamic connectivity system with three other recently published works that tackle the same problem. These systems are \dtree~\cite{chen2022dynamic}, \idtree~\cite{xu2024constant}, and \cdcname~\cite{fedorov2021scalable}.
%
\dtree and \cdcname both maintain a lossless representation of all edges in the graph and use different strategies to search these edges for a replacement whenever a spanning forest deletion occurs.
To compare against another sketch-based dynamic connectivity baseline, we also developed an implementation of \gibb. Prior to this work no implementations of a sketch-based dynamic connectivity algorithms existed, to the best of our knowledge.

\myparagraph{\dtree}
The \dtree algorithm~\cite{chen2022dynamic} maintains a low-diameter spanning forest of the graph. To find a replacement edge they simply perform BFS on the smaller component of the split caused by a tree edge deletion. During other operations they adjust the spanning forest such that the size of any smaller component formed by a future tree edge deletion will be low. This heuristic has no theoretical guarantees but has good performance in practice.
Our experiments use the original Python implementation of \dtree~\cite{chen2022dynamic}.

\myparagraph{\idtree}
\idtree~\cite{xu2024constant} improves upon \dtree with several optimizations, including a selectively rebuilt union-find/disjoint-set data structure. Additionally, the authors re-implement and optimize \dtree in C++. The authors provide several versions of their code with varying amounts of optimizations enabled to perform regression testing. We use the \idtree version of their code. 

\myparagraph{\cdcname}
\cdcname~\cite{fedorov2021scalable} is an implementation of a concurrent generalization of the dynamic connectivity algorithm of Holm \etal~\cite{holm2001poly}.
\cdcname implements a non-blocking single-writer Euler tour tree as a subroutine for updating trees edges in the Holm \etal algorithm. When a spanning forest edge is deleted, it applies a fine-grained lock to the component that the edge was part of, and searches its edge store for a replacement. Non-tree edge updates are processed concurrently and lock-free, and a non-tree edge addition which is being processed concurrently with a tree edge deletion may also be selected as its replacement.

The authors of \cdcname~\cite{fedorov2021scalable} provide implementations of several variants of their algorithm, all developed in Kotlin for JVM.
In our experiments we use the default implementation of their algorithm (referred to as "9: our algorithm" in their experimental section). We chose this version because it is consistently among the top performers in their experiments. This implementation uses the version of their algorithm with fine-grained locking, non-blocking non-tree edge updates, and non-blocking connectivity queries.
Since updates and queries are processed concurrently by their system, we are unable to report separate measurements for updates per second and queries per second, so we report operations per second where operations are updates or queries.

\myparagraph{\Gibb Implementation}
Our implementation of \Gibb~\cite{gibb2015dynamic} uses Euler tour trees implemented with skip lists. The only difference from their algorithm is we use \cameosketch~\cite{landscape} rather than the sketch of Cormode~\etal~\cite{cormode2014unifying}.

\myparagraph{Handling Out-Of-Memory (OOM) Errors}
As \dtree, \cdcname, and \idtree are not designed for dense graphs, they run out of main memory space on our machine for certain large inputs. As such, we don't report numbers for these instances.

\myparagraph{Measuring Query Performance}
For \sysname we include the time spent to flush the buffer of updates from the \speculative buffering technique in the total query time.
For \cdcname we can not directly measure query performance because queries happen concurrently with updates.

\section{Additional Experimental Results}\label{app:results}

\subsection{Isolated Updates Are Rare}\label{app:isolated} 
We ran \gibb on several datasets to get a sense of how common isolated updates are. On the (sparse) Twitter graph, $7\%$ of updates are isolated, while for the (dense) Kron-16 graph , $0.006\%$ of updates are isolated. 
Recall that each isolated update may produce $O(\log \nodesize)$ isolated components. In \gibb, processing each isolated component requires $O(\log^3 \nodesize)$ time. We found that a vast majority of isolated updates produce a small number (2 or less) of isolated components. Therefore the vast majority of isolated updates don't require the full $O(\log^4 \nodesize)$ time.

\subsection{\Speculative Buffering Tradeoff}\label{app:update_buffer_experiment}
The \speculative update buffering strategy described in Section~\ref{sec:update_buffer} has the potential to increase ingestion throughput when there are many normal updates. However, this increase may come at the cost of additional work when there are many isolated updates (because these require reverting earlier work) and high query latencies (because partially-full buffers must be flushed before the query can be computed). We tested this tradeoff by running \sysname on both the standard and fixed-forest versions of the Kron-16 and Twitter datasets, varying the size of the update buffers between 1 update and 1000 updates. This allows us to measure the buffer size effect on dense graphs (Kron-16), sparse graphs (Twitter), streams with many isolated updates (standard streams), and streams with few isolated updates (fixed-forest streams).


\begin{figure}
    \centering
    \includegraphics[width=\linewidth, trim={0cm, 1cm, 0cm, 0cm}]{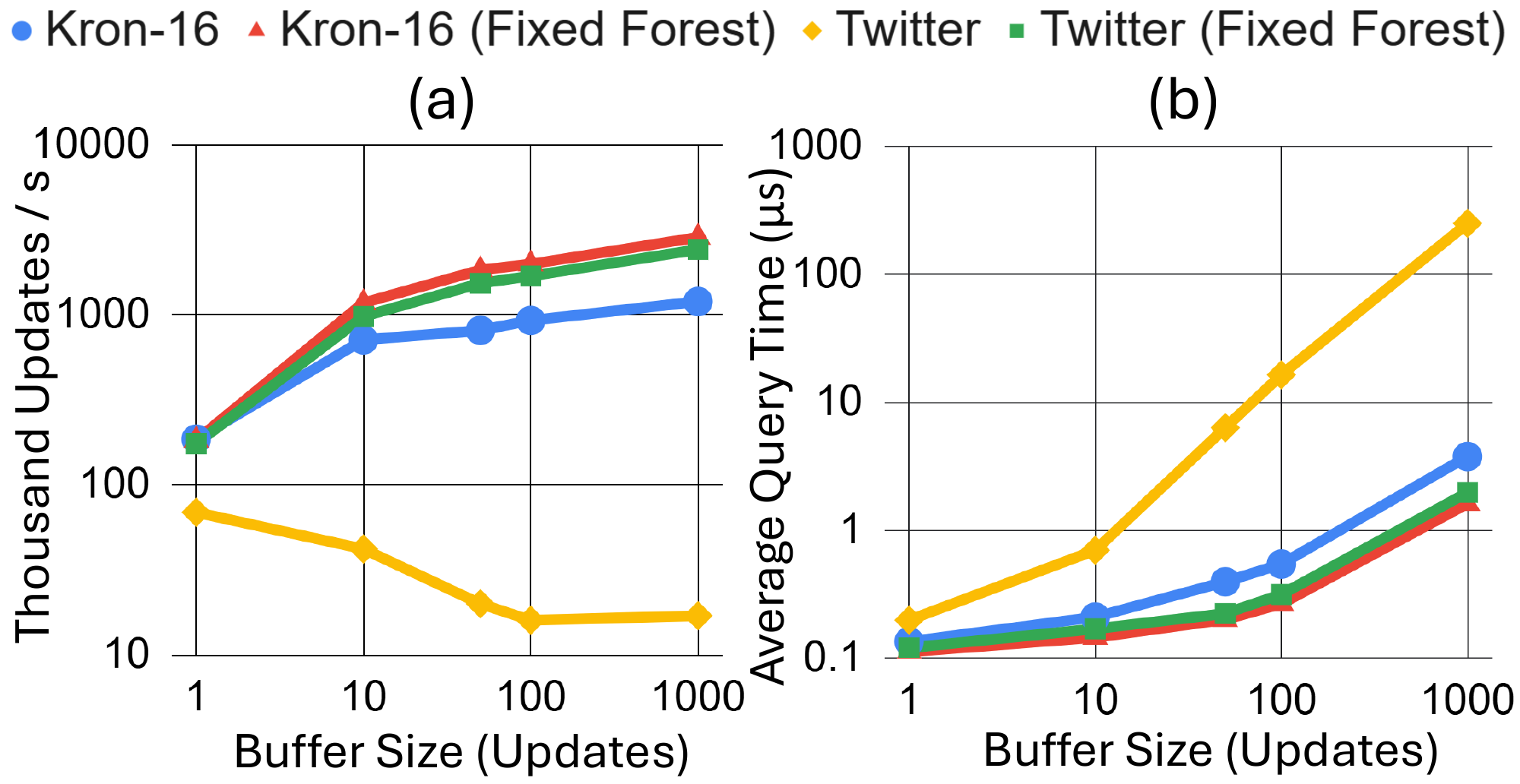}
    \caption{\small
    For graphs with few \treeedge updates (i.e., dense or fixed-forest graphs), increasing the speculative buffer size greatly improves ingestion throughput at the cost of query latency. For graphs with many \treeedge updates, larger speculative buffers degrade both ingestion throughput and query latency.
    }
    \label{fig:buffer}
\end{figure}

The results of this experiment are summarized in Figure~\ref{fig:buffer}. We see that increasing buffer size leads to increased update throughput on all datasets except for standard Twitter, which is sparse and has so many isolated updates that larger buffer sizes practically guarantee that most sketch update work will need to be reverted, perhaps even multiple times. For other datasets, isolated updates are not frequent enough to significantly lower the ingestion rate.

Unsurprisingly, increasing buffer size increases query latency for all datasets because larger buffers require more work to flush before a query can be answered. Again, this performance decrease is most striking for the standard Twitter graph, because there are likely to be many isolated updates to process before a query can be answered, and this leads to the same increases in processing cost that we observe in ingestion throughput. For all other datasets, while the query latency increases with buffer size, the magnitude of the increase is acceptable for most applications (because the baseline query latency is so low).

\end{document}